\documentclass[letterpaper, 11pt]{article}

\usepackage{amsmath, amssymb, amsthm, amsfonts, bbm}
\usepackage{graphicx}
\usepackage{enumerate}
\usepackage{palatino}
\usepackage[usenames,dvipsnames]{color}
\usepackage{units}
\usepackage{tikz}
\usepackage[utf8]{inputenc}
\usepackage[T1]{fontenc}
\usepackage[top=1in, bottom=1in, left=1in, right=1in]{geometry}
\usepackage{algorithm}
\usepackage{algpseudocode}
\usepackage{enumitem}
\usepackage{varwidth}
\usepackage[bookmarks=false]{hyperref}

\usepackage{xargs}                      
\usepackage[colorinlistoftodos,prependcaption,textsize=footnotesize]{todonotes}

\usetikzlibrary{positioning,snakes,mindmap,calc,fit,arrows,shapes}

\tikzstyle{box} = [shape=rectangle,draw=black,thick]

\newcommand{\eps}{\epsilon}

\newcommand{\poly}{\mathrm{poly}}

\newtheorem{theorem}{Theorem}[section]

\newtheorem{property}{Property}[section]

\newtheorem{lemma}[theorem]{Lemma}

\newtheorem{definition}[theorem]{Definition}

\newtheorem{remark}[theorem]{Remark}

\newcommand{\wt}{\mathrm{wt}}

\newcommand{\pienc}{\Pi_\mathrm{enc}^\mathrm{oblivious}}
\newcommand{\piblk}{\Pi_\mathrm{blk}}
\newcommand{\piblkrand}{\Pi_\mathrm{blk}}
\newcommand{\piencrand}{\Pi_\mathrm{enc}^\mathrm{random}}
\newcommand{\widt}{\widetilde}

\newcommand{\spka}{\mathsf{speak}_A}
\newcommand{\spkb}{\mathsf{speak}_B}

\newcommand{\E}{\mathbb{E}}
\newcommand{\crateless}{\mathcal{C}^\mathsf{rateless}}
\newcommand{\curerr}{\mathsf{err}}

\newcommand{\malAB}{\mathsf{mal}_\mathrm{AB}}
\newcommand{\malA}{\mathsf{mal}_\mathrm{A}}
\newcommand{\malB}{\mathsf{mal}_\mathrm{B}}
\newcommand{\Niter}{N_\mathsf{iter}}
\newcommand{\Nmal}{N_\mathsf{mal}}
\newcommand{\Ninv}{N_\mathsf{inv}}
\newcommand{\Nsound}{N_\mathsf{sound}}

\newcommand{\cinv}{C_\mathsf{inv}}
\newcommand{\cmal}{C_\mathsf{mal}}
\newcommand{\synca}{\mathsf{sync}_A}
\newcommand{\syncb}{\mathsf{sync}_B}

\newcommand{\rcvsyncb}{\widt{\mathsf{sync}}_B}
\newcommand{\inv}{\mathsf{inv}}
\newcommand{\shrand}{\mathsf{str}}
\newcommand{\shpseudo}{\mathsf{str}_\mathrm{stretch}}
\newcommand{\locrand}{\mathsf{str}^\mathrm{loc}}
\newcommand{\shremain}{\mathsf{str}'}
\newcommand{\chash}{\mathcal{C}^\mathsf{hash}}
\newcommand{\cexch}{\mathcal{C}^\mathsf{exchange}}

\newcommand{\nalt}{\mathsf{alt}}
\newcommand{\mcC}{\mathcal{C}}
\newcommand{\contenc}{\mathsf{Enc}}
\newcommand{\contdec}{\mathsf{Dec}}
\newcommand{\Unif}{\mathrm{Unif}}
\newcommand{\ctrlinfo}{\mathsf{ctrl}}
\newcommand{\cwt}{\mathsf{wt}}
\newcommand{\corrupt}{\mathsf{Corrupt}}

\def \todotoggle{0}
\newcommand{\mytodo}[1]{\ifnum\todotoggle=1{#1}\fi}

\newcommandx{\unsure}[2][1=]{\mytodo{\todo[linecolor=blue,
backgroundcolor=blue!25,bordercolor=blue,#1]{#2}}}
\newcommandx{\change}[2][1=]{\mytodo{\todo[linecolor=cyan,
backgroundcolor=cyan!25,bordercolor=cyan,#1]{#2}}}
\newcommandx{\info}[2][1=]{\mytodo{\todo[linecolor=green,
backgroundcolor=green!25,bordercolor=green,#1]{#2}}}
\newcommandx{\improve}[2][1=]{\mytodo{\todo[linecolor=red,backgroundcolor=red!25
,bordercolor=red,#1]{#2}}}
\newcommandx{\thiswillnotshow}[2][1=]{\todo[disable,#1]{#2}}

\newcommand{\FullOrShort}{full}
\ifthenelse{\equal{\FullOrShort}{full}}{
  \newcommand{\fullOnly}[1]{#1}
  \newcommand{\shortOnly}[1]{}
}{

  \newcommand{\fullOnly}[1]{}
  \newcommand{\shortOnly}[1]{#1}
}

\begin{document}

\title{Bridging the Capacity Gap Between Interactive and One-Way Communication}
\author{
Bernhard Haeupler\thanks{Computer Science Department, Carnegie Mellon 
University. Research supported in part by NSF grant CCF-1527110 and the 
NSF-BSF grant ``Coding for Distributed Computing.''
} \\Carnegie Mellon University\\haeupler@cs.cmu.edu\and
Ameya Velingker \thanks{Computer Science Department, Carnegie Mellon 
University. Part of this work was done while the author was visiting 
the Simons Institute for the Theory of Computing, Berkeley, CA. Research 
supported in part by NSF grant CCF-0963975.} \\Carnegie Mellon University\\avelingk@cs.cmu.edu
}

\date{}

\maketitle

\thispagestyle{empty}

\begin{abstract}
We study the communication rate of coding schemes for interactive communication 
that transform any two-party interactive protocol into a protocol that is 
robust to noise. 

Recently, Haeupler~\cite{Haeupler14} showed that if an $\epsilon > 0$ fraction
of transmissions are corrupted, adversarially or randomly, then it is possible to 
achieve a communication rate of $1 - \widetilde{O}(\sqrt{\epsilon})$. 
Furthermore, Haeupler conjectured that this rate is optimal for general input 
protocols. This stands in contrast to the classical setting of one-way 
communication in which 
error-correcting codes are known to achieve an optimal communication rate of 
$1 - \Theta(H(\epsilon)) = 1 - \widt{\Theta}(\epsilon)$. 

In this work, we show that the quadratically smaller rate loss of the one-way 
setting can also be achieved in interactive coding schemes for a very natural 
class of input protocols. We introduce the notion of \emph{average message
length}, or the average number of bits a party sends before receiving a reply,
as a natural parameter for measuring the level of interactivity in a protocol. 
Moreover, we show that any protocol with average message length $\ell = \Omega(\poly(1/\epsilon))$ can be
simulated by a protocol with optimal communication rate $1 - \Theta(H(\epsilon))$ over an oblivious
adversarial channel with error fraction $\epsilon$. Furthermore, under the 
additional assumption of access to public shared randomness, the optimal 
communication rate is achieved \emph{ratelessly}, i.e., the 
communication rate adapts automatically to the actual error rate 
$\epsilon$ without having to specify it in advance.
  
This shows that the capacity gap between one-way and interactive communication 
can be bridged even for very small (constant in $\epsilon$) average message
lengths, which are likely to be found in many applications.

\end{abstract}





\newpage

\setcounter{page}{1}

\section{Introduction}

In this work, we study the communication rate of coding schemes for interactive 
communication that transform any two-party interactive protocol into a protocol 
that is robust to noise.

\subsection{Error-Correcting Codes}
The study of reliable transmission over a noisy channel was pioneered by Shannon's work in the 1940s. He and others showed that error-correcting codes allow one to add redundancy to a message, thereby transforming the message into a longer sequence of symbols, such that one can recover the original message even if some errors occur. This allows fault-tolerant 
transmissions and storage of information. Error-correcting codes have since permeated most modern computation and communication technologies.

One focus of study has been the precise tradeoff between redundancy and fault-tolerance. In particular, if one uses an error-correcting code that encodes a binary message of length $k$ into a sequence of $n$ bits, 
then the \emph{communication rate} of the code is said to be $k/n$. One 
wishes to make the rate as high as possible. Shannon showed that for the random binary symmetric channel (BSC) 
with error probability $\epsilon$ the (asymptotically) best achievable rate is $C = 1-H(\epsilon)$, where $H(\epsilon) = 
-\epsilon\log_2 {\epsilon} - (1-\epsilon)\log_2 (1-\epsilon)$ denotes 
the \emph{binary entropy function}.

Another realm of interest is the case of \emph{adversarial} errors. In this 
case, the communication channel corrupts at most an $\epsilon$ fraction of the 
total number of bits that are transmitted. Moreover, one wishes to allow the 
receiver to correctly decode the message in the presence of any such error 
pattern. The work of Hamming shows that one can achieve a communication rate of 
$R = 1-\Theta(H(\epsilon))$, in particular, the so-called Gilbert-Varshamov 
bound of $1-H(2\epsilon) > 1 - 2H(\epsilon)$. Determining the optimal rate, or even just the  constant hidden by the asymptotic $\Theta(H(\epsilon))$ term, remains a major open question.

\subsection{Interactive Communication}
The work of Shannon and Hamming applies to the problem of \emph{one-way 
communication}, in which one party, say Alice, wishes to send a message to 
another party, say Bob. However, in many applications, underlying (two-party) 
communications are \emph{interactive}, i.e., Bob's  response to Alice may be 
based on what he received from her previously and vice versa. 
As in the case of one-way communication, one wishes to make such interactive communications robust to noise 
by adding some redundancy. 

At first sight, it seems plausible that one could  use  error-correcting 
codes to encode each round of communication separately. However, this does not
work correctly because the channel might corrupt the codeword of one such round of communication entirely and
as a result derail the entire future conversation. With the naive approach being insufficient, it is not obvious whether it is possible at all to encode interactive protocols in a way that can tolerate some small constant fraction of errors in an 
interactive setting. Nonetheless, Schulman~\cite{Schulman92, Schulman93, 
Schulman96} showed that this is possible and numerous follow-up works over the past 
several years have led to a drastically better understanding of 
error-correcting coding schemes for interactive communications
\shortOnly{(see Appendix~\ref{app:relworks})}.

\subsection{Communication Rates of Interactive Coding Schemes} \label{sec:introICrate}
Only recently, however, has this study led to results shedding light on the tradeoff between the achievable communication rate for a given error fraction or amount of noise.

Kol and Raz~\cite{KR13} gave a communication scheme for random errors that achieves a communication rate of $1-O(\sqrt{H(\epsilon)})$ for any alternating protocol, where $\epsilon>0$ is the error rate. \cite{KR13} also developed powerful tools to prove upper bounds on the communication rate. Haeupler~\cite{Haeupler14} showed communication schemes that achieve a communication rate of $1-O(\sqrt{\epsilon})$ for any oblivious adversarial channel, including random errors, as well as a communication rate of $1-O(\sqrt{\epsilon \log\log(1/\epsilon)})$ for any fully adaptive adversarial channel. These results apply to alternating protocols as well as adaptively simulated non-alternating protocols (see \cite{Haeupler14} for a more detailed discussions). Lastly, given \cite{KR13}, Haeupler conjectured these rates to be optimal for their respective settings. Therefore, there is an almost quadratic gap between the conjectured rate achievable in the interactive setting and the $1-\Theta(H(\epsilon))$ rate known to be optimal for one-way communications.

\subsection{Results}

In this paper, we investigate this communication rate gap. In particular, we show 
that for a natural and large class of protocols this gap disappears. Our primary 
focus is on protocols for \emph{oblivious adversarial} channels. Such a channel 
can corrupt any $\epsilon$ fraction of bits that are exchanged in the execution 
of a protocol, and the simulation is required to work, with high probability, 
for any such error pattern. This is significantly stronger, more interesting, 
and, as we will see, also much more challenging than the case of independent 
random errors. We remark that, in contrast to a \emph{fully adaptive 
adversarial} channel, the decision whether an error happens in a given round is 
not allowed to depend on the transcript of the execution thus far. This seems to 
be a minor but crucially necessary restriction (see also Section 
\ref{sec:conceptual}).


As mentioned, the conjectured optimal communication rate of $1-O(\sqrt{\epsilon})$ for the oblivious adversarial setting is worse than the $1-O(H(\epsilon))$ communication rate achievable in the one-way communication settings. However, the conjectured upper bound seems to be tight mainly for ``maximally interactive'' protocols, i.e., protocols in which the party that is sending
bits changes frequently. In particular, \emph{alternating} protocols, in which
Alice and Bob take turns sending a single bit, seem to require the most
redundancy for a noise-resilient encoding. On the other hand, the usual
one-way communication case in which one party just sends a single message
consisting of several bits is an example of a ``minimally
interactive'' protocol. It is a natural question to consider
what the tradeoff is between achievable communication rate and the level of interaction that takes
place. In particular, most natural real-world protocols are rarely ``maximally
interactive'' and could potentially be simulated with communication rates going well beyond
$1-O(\sqrt{\epsilon})$. We seek to investigate this possibility.

Our first contribution is to introduce the notion of \emph{average 
message length} as a natural measure of the interactivity of a protocol in the 
context of analyzing communication rates. Loosely speaking, the average message 
length of an $n$-round protocol corresponds to the average number of bits a 
party sends before receiving a reply from the other party. A lower average 
message length roughly corresponds to more interactivity in a protocol, e.g., 
a maximally interactive protocol has average message length 1, while a 
one-way protocol with no interactivity has average message length $n$. The 
formal definition of average message length appears as 
Definition~\ref{def:avgmsg} in Section~\ref{subsec:msglen}.

Our second and main contribution in this paper is to show that for protocols with
an average message length of at least some constant in $\epsilon$ (but independent of
the number of rounds $n$) one can go well beyond the $1-\Theta(\sqrt{\epsilon})$ communication rate achieved by 
\cite{Haeupler14} for channels with oblivious adversarial errors. In fact, we 
show that for such protocols one can actually achieve a communication rate of 
$1-\Theta(H(\epsilon))$, matching the communication rate for one-way 
communication up to the (unknown) constant in  the $H(\epsilon)$ term.

\begin{theorem} \label{thm:mainoblivious}
For any $\eps > 0$ and any $n$-round interactive protocol $\Pi$ 
with average message
length $\ell = \Omega(\mathrm{poly}(1/\epsilon))$, it is possible to 
encode $\Pi$ into a protocol over the same alphabet which, with probability at 
least $1 - \exp(-n \epsilon^6)$, simulates 
$\Pi$ over an oblivious adversarial channel with an $\epsilon$ fraction of errors while 
achieving a communication rate of $1-\Theta(H(\epsilon)) = 1-\Theta(\epsilon\log(1/\epsilon))$.
\end{theorem}

Under the (simplifying) assumption of \emph{public shared randomness}, our 
protocol can furthermore be seen to have the nice property of being 
\emph{rateless}. This means that the communication rate adapts automatically 
and only depends on the actual error rate $\eps$ without having to specify or 
know in advance what amount of noise to prepare for.

\begin{theorem} \label{thm:rateless}
 Suppose Alice and Bob have access to public shared randomness. For any 
$\epsilon' > 0$ and any $n$-round interactive protocol $\Pi$ with average 
message length $\ell = \Omega(\mathrm{poly}(1/\epsilon'))$, it is possible to 
encode $\Pi$ into protocol $\Pi_\mathrm{rateless}$ over 
the same alphabet such that for any \emph{true error rate} $\epsilon$, executing 
$\Pi_\mathrm{rateless}$ for $n(1 + O(H(\epsilon)) + 
O(\epsilon'\,\mathrm{polylog}(1/\epsilon')))$ rounds simulates $\Pi$ with 
probability at least $1 - \exp(-n\epsilon'^3)$.
\end{theorem}
We note that one should think of $\epsilon'$ in Theorem~\ref{thm:rateless} as 
chosen to be very small, in particular, smaller than the smallest amount of 
noise one expects to encounter. In this case, the communication rate of the 
protocol simplifies to the optimal $1 - O(H(\epsilon))$ for essentially any 
$\epsilon > \epsilon'$. The only reason for not choosing $\epsilon'$ too small 
is that it very slightly increases the failure probability. As an example, 
choosing $\epsilon'=o(1)$ suffices to get ratelessness for any constant 
$\epsilon$ and still leads to an essentially exponential failure probability. 
Alternatively, one can even set $\epsilon' = n^{-1/6}$ which leads to optimal 
communication rates even for tiny sub-constant true error fractions $\epsilon > 
n^{-0.2}$ while still achieving a strong sub-exponential failure probability of 
at most $\exp(-\sqrt{n})$.

\global\def\RelatedWorkSection{
Schulman was the first to consider the question of coding for 
interactive communication and showed that one can tolerate an adversarial error 
fraction of $\epsilon = 1/240$ with an unspecified constant 
communication rate~\cite{Schulman92, Schulman93, Schulman96}. Schulman's result 
also 
implies that for the easier setting of random errors, one can tolerate any
error 
rate bounded away from $1/2$ by repeating symbols multiple times. Since 
Schulman's seminal work, there has been a number of subsequent works pinning 
down the tolerable error fraction. For instance, Braverman and Rao~\cite{BR14} 
showed that any error fraction $\epsilon < 1/4$ can be tolerated in the realm 
of adversarial errors, provided that one can use larger alphabet sizes, and 
this bound was shown to be optimal. A series of subsequent works~\cite{BE14, 
GH14, GHS14, EGH15, FGOS15} worked to determine the error rate region under 
which non-zero communication rates can be obtained for a variety of models, 
e.g., adversarial errors, random errors, list-decoding, adaptivity, and channels 
with feedback. Unlike the initial coding schemes of \cite{Schulman96} and \cite{BR14} that relied on tree codes and as a result required exponential time computations, many of the newer coding schemes are 
computationally efficient~\cite{BK12, BN13, BKN14, GMS14, GH14}. All these results achieve 
small often unspecified constant communication rate of $\Theta(1)$ which is fixed and independent of amount of noise. Only the works of \cite{KR13} and \cite{Haeupler14}, which are already discussed above in Section \ref{sec:introICrate} achieve a communication rate approaching $1$ for error fractions going to zero.


%
}

\fullOnly{
\subsection{Further Related Works}
\RelatedWorkSection
}
\shortOnly{
\section{Preliminaries}
The definitions and notation for the coding for interactive communication setting used throughout the paper are standard and are stated in detail in Appendix~\ref{app:prelim}.
}

\global\def\Prelim{
An \emph{interactive protocol} $\Pi$ consists of communication performed by 
two parties, Alice and Bob, over a channel with alphabet $\Sigma$. Alice has 
an input $x$ and Bob has an input $y$, and the protocol consists of $n$ 
\emph{rounds}. During each 
round of a protocol, each party decides whether to listen or transmit a symbol 
from $\Sigma$, based on his input and the player's \emph{transcript} thus 
far. Alice's \emph{transcript} is defined as a tuple of symbols from $\Sigma$, 
one for each round that has occurred, such that the $i^\text{th}$ symbol is 
either (a.) the symbol that Alice sent during the $i^\text{th}$ round, if she 
chose to transmit, or (b.) the symbol that Alice received, otherwise.

Moreover, protocols 
can utilize 
\emph{randomness}. In the case 
of \emph{private randomness}, each party is given its own infinite string of 
independent uniformly random bits as part of its input. In the case of \emph{shared randomness}, 
both 
parties have access to a common infinite random string during each round. In 
general, our protocols will utilize private randomness, unless otherwise 
specified.

In a \emph{noiseless} setting, we can assume that in any round, exactly one 
party speaks and one party listens. In this case, the listening party simply 
receives the symbol sent by the speaking party.

The \emph{communication order} of a protocol refers to the order in which Alice 
and Bob choose to speak or listen. A protocol is \emph{non-adaptive} if the 
communication order is fixed prior to the start of the protocol, in which case, 
whether a party transmits or listens depends only on the round number. A simple 
type of non-adaptive protocol is an \emph{alternating} protocol, in which one 
party transmits during odd numbered rounds, while the other party transmits 
during even numbered rounds. On the other hand, an \emph{adaptive} protocol is 
one in which the communication order is not fixed prior to the start; 
therefore, the communication order can vary depending on the transcript of the 
protocol. In particular, each party's decision whether to speak or listen 
during a round will depend on his input, randomness, as well as the transcript 
of the protocol thus far.

For an $n$-round protocol over alphabet $\Sigma$, one can define an associated
\emph{protocol tree} of depth $n$. The protocol tree is a rooted tree in 
which each non-leaf node of the tree has 
$|\Sigma|$ children, and the outgoing edges are labeled by the elements of 
$\Sigma$. Each non-leaf node is owned by some player, and the owner of the node 
has a \emph{preferred} edge that emanates from the node. The preferred edge is 
a function of the owner's input and any randomness that is allowed. Also, leaf 
nodes of the protocol tree correspond to ending states.

A proper execution of the protocol corresponds to the unique path from the root 
of the protocol tree to a leaf node, such that each traversed edge is the 
preferred edge of the parent node of the edge. In this case, each edge along 
the path can be viewed as a successive round in which the owner of the parent 
node transmits the symbol along the edge.

An example of a protocol tree is shown in Figure~\ref{fig:protocoltree}.

\begin{figure}[h]
 \centering
 \includegraphics[width=5in]{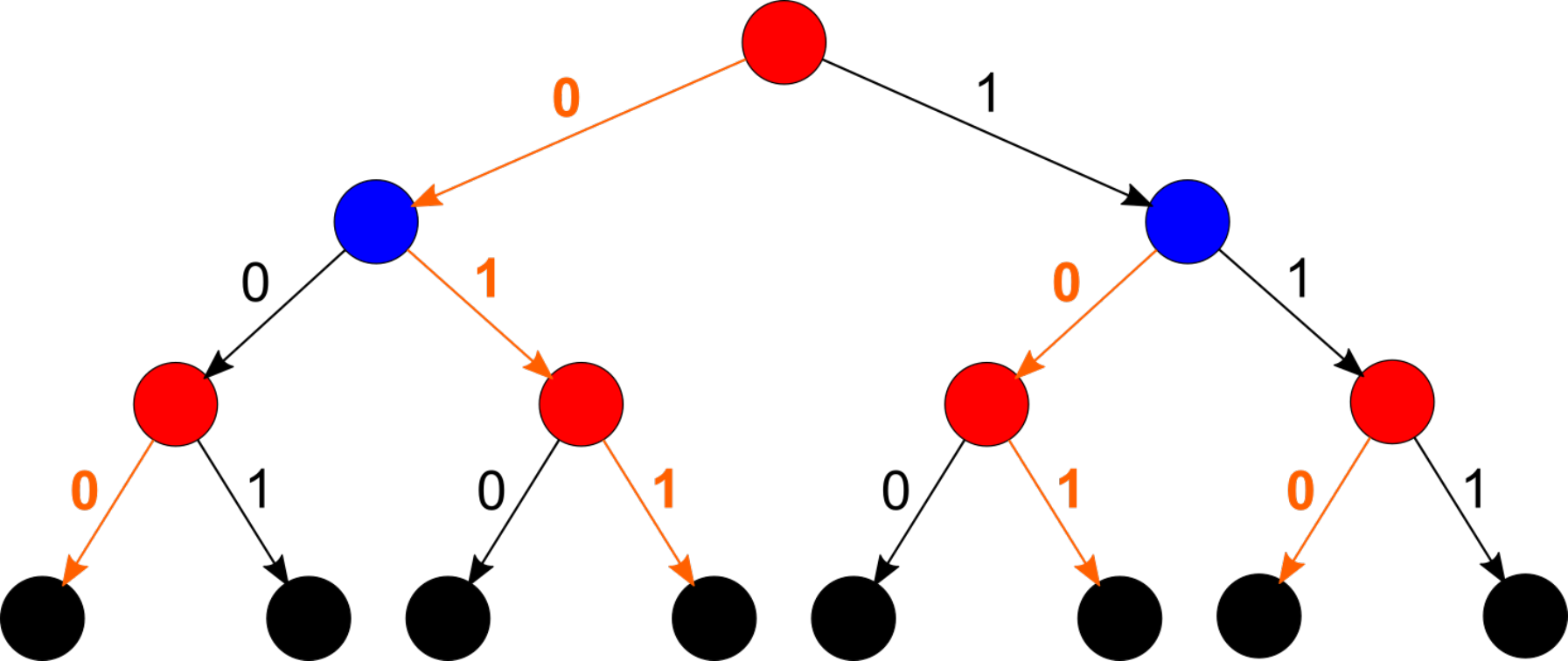}
 \caption{An example of a protocol tree for a 3-round interactive protocol. 
Nodes owned by Alice are colored red, while those owned by Bob are colored 
blue. Note that Alice always speaks during the first and third rounds, while 
Bob 
speaks during the second round. The orange edges are the set of preferred edges 
for some choice of inputs of Alice and Bob. In this case, a proper execution of 
the protocol corresponds to the path ``011.''}
 \label{fig:protocoltree}
\end{figure}
}

\global\def\CommChannels{
For our purposes, the communication between the two parties occurs over a 
\emph{communication channel} that delivers a possibly corrupted version of the 
symbol transmitted by the sending party. In this work, transmissions will be 
from a \emph{binary} alphabet, i.e., $\Sigma = \{0,1\}$.

In a \emph{random error channel}, each transmission occurs over a binary 
symmetric channel with crossover probability $\epsilon$. In other words, in 
each round, if only one party is speaking, then the transmitted bit gets 
corrupted with probability $\epsilon$.

This work mainly considers the \emph{oblivious adversarial channel}, in which 
an adversary gets to corrupt at most $\epsilon$ fraction of the total number of 
rounds. However, the adversary is restricted to making his decisions prior to 
the start of the protocol, i.e., the adversary must decide which rounds to 
corrupt independently of the communication history and randomness used by Alice 
and Bob. For each round that the adversary decides to corrupt, he can either 
commit a \emph{flip} error or \emph{replace} error. Suppose a round has one 
party that speaks and one party that listens. Then, a flip error means that 
the listening party receives the opposite of the bit that the transmitting 
party sends. On the other hand, a replace error requires the adversary to 
specify a symbol $\alpha\in\Sigma$ for the round. In this case, the listening 
party receives $\alpha$ regardless of which symbol was sent by the transmitting 
party.

An {adaptive adversarial channel} allows an adversary 
to corrupt at most 
$\epsilon$ fraction of the total number of rounds. However, in this case, the 
adversary does not have to commit to which rounds to corrupt prior to the start 
of the protocol. Rather, the adversary can decide to corrupt a round based on 
the communication history thus far, including what is being sent in the 
current round. Thus, in any round that the adversary chooses to corrupt in 
which one party transmits and one party receives, the adversary can make the 
listening party receive any symbol of his choice.

Note that we have not yet specified the behavior for rounds in which both 
parties speak or both parties listen. Such rounds can occur for \emph{adaptive} 
protocols when the communication occurs over a noisy communication channel.

If both parties speak during a round, we stipulate that neither party receives 
any symbol during that round (since neither party is expecting to receive a 
symbol).

Moreover, we stipulate that in rounds during which both parties listen, the 
symbols received by Alice and Bob are unspecified. In other words, an 
arbitrary symbol may be delivered to each of the parties, and we require that 
the protocol work for any choice of 
received symbols. Alternatively, one can imagine that the adversary 
chooses arbitrary symbols for Alice and Bob to receive without this being 
counted as a corruption (i.e., a free corruption that is not counted toward 
the budget of $\epsilon$ fraction of corruptions). The reason for this model is 
to disallow the possibility of transmitting information by using silence. 
An extensive discussion on the appropriateness of this error model can be 
found in~\cite{GHS14}.
}

\fullOnly{
\section{Preliminaries}
\Prelim
\subsection{Communication Channels} \label{subsec:commchannel}
\CommChannels
}

\section{Average Message Length and Blocked Protocols}\label{subsec:msglen}
One conceptual contribution of this work is to introduce the notion of 
\emph{average message length} as a natural measure of the level of interactivity of a protocol. While this paper uses it only in the  
context of analyzing the optimal rate of interactive coding schemes, we believe that this notion and parametrization will also be useful in other settings, such as compression. Next, we define this notion formally.

\begin{definition}\label{def:avgmsg}
 The \emph{average message length} $\ell$ of an $n$-round interactive protocol
$\Pi$ is the minimum, over all paths in the protocol tree of $\Pi$, of the
average length in bits of a maximal contiguous block (spoken by a single
party) down the path.

More precisely, given any string $s\in\{0,1\}^n$, there exist integer message lengths $l_0, \ldots, l_k>0$ such that along the path of $\Pi$ given by $s$ one player (either Alice or Bob) speaks between round $1 + \sum_{j < i} l_j$ and round $\sum_{j \leq i} l_j$ for even $i$ while the other speaks during the remaining intervals, i.e., those for odd $i$. We then define $\ell_s$ to be the average of these message lengths $l_0, \ldots, l_k$ and define the \emph{average
message length} of $\Pi$ to be minimum over all possible inputs, i.e., $\ell = \min_{s\in\{0,1\}^n} \ell_s$.
\end{definition}

An alternate characterization of the amount of interaction in a protocol
involves the number of alternations in the protocol:

\begin{definition}
 An $n$-round protocol $\Pi$ is said to be \emph{$k$-alternating} if any path 
in the protocol tree of $\Pi$ can be divided into at most $k$ blocks of 
consecutive rounds such that only one person (either Alice or Bob) speaks 
during each block.

More precisely, $\Pi$ is \emph{$k$-alternating} if, given any string $s \in
\{0,1\}^n$, there exist $k' \leq k$ integers $r_0, r_1, \dots, r_{k'}$ with $0 = r_0 < \dots < 
r_{k'} = n$, such that along the path of $\Pi$ given by $s$, only one player 
(either Alice or Bob) speaks for rounds $r_i + 1, \dots, r_{i+1}$ for any 
$0\leq i < k'$.
\end{definition}

It is easy to see that the two notions are essentially equivalent, as an $n$-round protocol
with average message length $\ell$ is an $(n/\ell)$-alternating protocol, and a
$k$-alternating $n$-round protocol has average message length $n/k$.
Note that an $n$-round \emph{alternating} protocol has average message
length $1$, while a \emph{one-way} protocol has average message length $n$. 
The average message length can thus be seen as a natural measure for the interactivity of a protocol.

We emphasize that the average message length definition does not require 
message lengths to be uniform along any path or across paths. In particular, 
this allows for the length of a response to vary depending on what was 
communicated before, e.g., the statement the other party has just made---a 
common phenomenon in many applications. Taking as an example real-world 
conversations between two people, responses to statements can be as short as a 
simple ``I agree'' or much longer, depending on what the conversation 
has already covered and what the opinion or input of the receiving party is. 
Thus, a sufficiently large average message length roughly states that while the 
$i^\text{th}$ response of a person can be short or long depending on the history 
of the conversation, no sequence of responses can lead to two parties going back 
and forth with super short statements for too long a period of time. This 
flexibility makes the average message length a highly applicable parameter that 
is reasonably large in most settings of interest. We expect it to be a very 
useful parametrization for questions going beyond the communication rate 
considered here.

However, the non-uniformity of protocols with an average message length bound 
can make the design and analysis of protocols somewhat harder than one would 
like. Fortunately, adding some dummy rounds of communication in a 
simple procedure we call \emph{blocking} allows us to transform any protocol 
with small number of alternations into a much more regularly structured protocol 
which we refer to as \emph{blocked}.

\begin{definition}
 An $n$-round protocol $\Pi$ is said to be \emph{$b$-blocked} if for any 
$1\leq j \leq \lceil n/b \rceil$, only one person (either Alice or Bob) speaks 
during all rounds $r$ such that $(j-1)b < r \leq jb$.
\end{definition}

\begin{lemma}\label{lem:blocking}
Any $n$-round $k$-alternating protocol $\Pi$ can be simulated by a $b$-blocked protocol $\Pi'$ that consists of at most $n+kb$ rounds.
\end{lemma}

\global\def\ProofofBlocking{
\begin{proof}[Proof of Lemma~\ref{lem:blocking}]
 Consider the protocol tree of $\Pi$, where each node corresponds to a state of 
the protocol (with the root as the starting state) and each node has at most 
two edges leaving from it (labeled `0' and `1'). Moreover, each node is colored 
one of two colors depending on whether Alice or Bob speaks next in the 
corresponding state, and the edges emanating from the node are colored the 
same. The leaves of the protocol tree are terminating states of 
the protocol, and one can view any (possibly corrupted) execution of the 
protocol as a path from the root to a leaf of the tree, where the edge taken 
from any node indicates the bit that is transmitted by the sender from the 
corresponding state.

Now, consider any path down the protocol tree. We can group the edges of the 
path into maximal groups of consecutive edges of the same color. Now, if any 
group of edges contains a number of edges that is not a multiple of $b$, then 
we add some dummy nodes (with edges) in the middle of the group so that the new 
number of edges in the group is the next largest multiple of $b$. It is clear 
that if we do this for every path down the original protocol tree, then the 
resulting protocol tree will correspond to a protocol $\Pi'$ that is
$b$-blocked 
and simulates $\Pi$ (i.e., each leaf of $\Pi'$ corresponds to a leaf of $\Pi$).

Moreover, note that the number of groups of edges is at most $k$, since $\Pi$ 
is $k$-alternating. Also, the number of dummy nodes we add in each group is at 
most $b$. It follows that the number of nodes (and edges) down any original 
path of $\Pi$ has increased by at most $kn$ in $\Pi'$. Thus, the desired claim 
follows.
\end{proof}
}

\fullOnly{\ProofofBlocking}
\shortOnly{\noindent The proof of Lemma~\ref{lem:blocking} is straightforward 
and appears in Appendix~\ref{app:proofs}.}

\section{Warmup: Interactive Coding for Random Errors}\label{sec:random}
As a warmup for the much more difficult adversarial setting, we first consider 
the setting of random errors, as this will illustrate several ideas including 
blocking, the use of error-correcting codes, and how to incorporate those with 
known techniques in coding for interactive communication.

In this section, we suppose that each transmission of Alice and Bob occurs over a binary symmetric channel with an $\epsilon$ probability of corruption. Recall that we wish to encode an $n$-round protocol $\Pi$ into a protocol $\piencrand$ such that with high probability over the communication channel, execution of $\piencrand$ robustly simulates $\Pi$. By \cite{Haeupler14}, it is known that 
one can achieve a communication rate of $1-O(\sqrt{\epsilon})$. In this section, we show how to go beyond the rate of $1-O(\sqrt{\epsilon})$ for protocols with at least a constant (in $\epsilon$) average message length.

\shortOnly{
}

\global\def\TrivialScheme{
The first coding scheme we present for completeness is a completely trivial and straight forward application of error correcting codes which works for \emph{non-adaptive} protocols $\Pi$ with a guaranteed \emph{minimum} message 
length. In particular, the coding scheme achieves a communication rate of $1-O(H(\epsilon))$ for non-adaptive protocols 
with minimum message length $\Omega((1/\epsilon)\log n)$.

In particular, we assume that $\Pi$ is a a non-adaptive $n$-round protocol with 
message lengths of size $b_1, b_2, \dots, b_k$, i.e., Alice sends $b_1$ bits, 
then Bob sends $b_2$ bits, and so on. Moreover, we assume that that $b_1, b_2, \dots, b_k \geq b$, where $b = \Omega((1/\epsilon)\log n)$ is the minimum message length.

Now, we can form the encoded protocol $\piencrand$ by simply having the 
transmitting party replace its intended message in $\Pi$ (of $b_i$ bits) 
with the encoding (of length, say, $b_i'$) of the message under an 
error-correcting code of minimum relative distance $\Omega(\epsilon)$ and
rate $1-O(H(\epsilon))$  and then transmitting the resulting codeword. The 
receiver then decodes the word according to the nearest codeword of the
appropriate error-correcting code.

Note that for any given message (codeword) of length $b_i'$, the expected 
number of
corruptions due to the channel is $\epsilon b_i'$. Thus, by Chernoff bound, 
the probability
that the corresponding codeword is corrupted beyond half the minimum distance 
of the relevant error-correcting code is
$e^{-\Omega(\epsilon b')} = n^{-\Omega(1)}$. Since $k = O(n/b) = 
O(n\epsilon/\log
n)$, the union bound implies that the probability that any of the $k<n$ messages
is corrupted beyond half the minimum distance is also $n^{-\Omega(1)}$. Thus, with
probability $1-n^{-\Omega(1)}$, $\piencrand$ simulates the original
protocol without error. Moreover, the overall communication rate is clearly $1 
- O(H(\epsilon))$ due to the choice of the error-correcting codes.

\begin{remark}
Note that the aforementioned trivial coding scheme has the disadvantage of 
working only for nonadaptive protocols with a certain \emph{minimum} message 
length, which is a much stronger assumption than average message length. In 
Section~\ref{subsec:randblock}, we show how to get around this problem by 
converting the input protocol to a \emph{blocked} protocol.

Another problem with the coding scheme is that the minimum message length is 
required to be $\Omega_\epsilon(\log n)$. This is in order to ensure that the 
probability of error survives a union bound, as the trivial coding scheme has 
no mechanism for recovering if a particular message gets corrupted. This also results in 
a success probability of only $1 - 1/\poly(n)$ instead of the $1 - \exp(n)$ one would like to 
have for a coding scheme. Section~\ref{subsec:randblock} shows how to rectify both problems by combining the 
reduced error probability of a error correcting code failing with any existing interactive coding scheme, such as \cite{Haeupler14}.
\end{remark}
}

\fullOnly{
\subsection{Trivial Scheme for Non-Adaptive Protocols with Minimum Message 
Length}
\TrivialScheme
}

\subsection{Coding Scheme for Protocols with 
Average Message Length of $\Omega(\log(1/\epsilon)/\epsilon^2)$} 
\label{subsec:randblock}

\fullOnly{In this section, we build on the trivial scheme discussed earlier to provide 
an improved coding scheme that handles any protocol $\Pi$ with an \emph{average
message length} of at least $\ell = \Omega (\log(1/\epsilon) / \epsilon^2)$.}

The first step will be to transform $\Pi$ into a protocol that is blocked. 
Note that the $\Pi$ is a $k$-alternating protocol, where $k = n/\ell = 
O(n\epsilon^2 / \log(1/\epsilon))$. Thus, by Lemma~\ref{lem:blocking}, we can 
transform $\Pi$ into a $b$-blocked protocol $\piblkrand$, for $b = 
\Theta(\log(1/\epsilon)/\epsilon)$, such that $\piblkrand$ simulates $\Pi$ 
and consists of $n_b = n+kb = n(1+O(\epsilon))$ rounds.

Now, we view $\piblk$ as a $q$-ary protocol with $n_b / b$ rounds, 
where $q = 2^b$. This can be done by grouping the symbols in each $b$-sized 
block as a single symbol from an alphabet of size $q$. Next, we can use the 
coding scheme of \cite{Haeupler14} in a blackbox manner to encode this $q$-ary 
protocol as a $q$-ary protocol $\Pi'$ with $\frac{n_b}{b}(1 + 
\Theta(\sqrt{\epsilon'}))$ rounds such that $\Pi'$ simulates $\Pi$ under oblivious random 
errors with error fraction $\epsilon'$ (i.e., each $q$-ary symbol is 
corrupted (in any way) with an independent probability of at most $\epsilon'$). We pick $\epsilon' = \epsilon^4$.

Finally, we transform $\Pi'$ into a \emph{binary} protocol $\piencrand$ as 
follows: We expand each $q$-ary symbol of $\Pi'$ back into a sequence of 
$b$ bits and then expand the $b$ bits into $b' > b$ bits using an 
error-correcting code. In particular, we use an error-correcting code 
$\mathcal{C}: \{0,1\}^b \to \{0,1\}^{b'}$ with block length $b' = b + 
(2c+\delta)\log^2 (1/\epsilon)$ and minimum distance $2c\log(1/\epsilon)$ for 
appropriate constants $c,\delta$ (such a code is guaranteed to exist by the 
Gilbert-Varshamov bound). Thus, $\piencrand$ is a $b'$-blocked binary protocol 
with $n_b\cdot \frac{b'}{b} (1+\Theta(\sqrt{\epsilon'})) = n(1 + 
O(\epsilon\log(1/\epsilon))$ rounds. Moreover, each $b'$-sized block of  
$\piencrand$ simply simulates each $q$-ary symbol of $\Pi'$ and the listening party simply decodes the received $b'$ bits to the nearest codeword of $\mathcal{C}$.

To see that $\piencrand$ successfully simulates $\Pi$ in the presence of 
random errors with error fraction $\epsilon$, observe that a $b'$-block is 
decoded incorrectly if and only if more than $d/2$ of the $b'$ bits are 
corrupted. By the Chernoff bound, the probability of such an event is $< 
\epsilon^4$ (for appropriate choice of $c,\delta$). Thus, since $\Pi'$ is known 
to simulate $\Pi$ under oblivious errors with error fraction $\epsilon^4$, it 
follows that $\piencrand$ satisfies the desired property.
\section{Conceptual Challenges and Key Ideas}\label{sec:conceptual}

In this section, we wish to provide some intuition for the difficulties in 
surpassing the $1-\Theta(\sqrt{\epsilon})$ communication rate for interactive 
coding when dealing with non-random errors. We do this because the adversarial 
setting comes with a completely new set of challenges that are somewhat subtle 
but nonetheless fundamental. As such, the techniques used in the previous 
section for interactive coding under random errors still provide a good 
introduction to some of the building blocks in the framework we use to deal with 
the adversarial setting, but they are not sufficient to circumvent the main 
technical challenges. Indeed, we show in this section that the adversarial 
setting inherently requires several completely new techniques to beat the 
$1-\Theta(\sqrt{\epsilon})$ communication rate barrier. 

We begin by noting that all existing interactive coding schemes encode the input 
protocol $\Pi$ into a protocol $\Pi'$ with a certain type of structure: There 
are some, \emph{a priori} specified, communication rounds which simulate rounds 
of the original protocol (i.e., 
result in a walk down the protocol tree of $\Pi$), while other rounds constitute 
\emph{redundant information} which is used for error correction. In the case of 
protocols that use hashing (e.g., \cite{Haeupler14}, \cite{KR13}), this is 
directly apparent in their description, as rounds in which hashes and control 
information are communicated constitute redundant information. However, this 
is also the case for all protocols based on tree codes (e.g., 
\cite{BR14,GHS14,GH14}): To see this, note that in such protocols, one can 
simply use an underlying tree code that is linear and systematic, with the 
non-systematic portion of the tree code then corresponding to redundant rounds.

We next present an argument which shows that, due to the above structure, no 
existing coding scheme can break the natural $1-\Omega(\sqrt{\epsilon})$ 
communication rate barrier, even for protocols with near-linear $o(n)$ average 
message lengths. This will also provide some intuition about what is required to 
surpass this barrier.

Suppose that for a (randomized) $n$-round communication protocol $\Pi$, the 
simulating protocol $\Pi'$ has the above structure and a communication rate of 
$1 -\epsilon'$. The simulation $\Pi'$ thus consists of exactly $N = n / 
(1-\epsilon')$ rounds. Note that, since every simulation must have at least $n$ 
non-redundant rounds, the fraction of redundant rounds in $\Pi'$ can be at most 
$\epsilon'$. Given that the position of the redundant rounds is fixed, it is 
therefore possible to find a window of $(\epsilon/\epsilon') N$ consecutive 
rounds in $\Pi'$ which contain at most $\epsilon N$ redundant rounds, i.e., an 
$\epsilon'$ fraction. Now, consider an oblivious adversarial channel that 
corrupts all the redundant information in the window along with a few extra 
rounds. Such an adversary renders any error correction technique useless, while 
the few extra errors derail the unprotected parts of the communication, thereby 
rendering essentially all the non-redundant information communicated in this 
window useless as well---all while corrupting essentially only $\epsilon N$ 
rounds in total. This implies that in the remaining $N - (\epsilon/\epsilon') 
N$ communication rounds outside of this window, there must be at least $n$ 
non-redundant rounds in order for $\Pi'$ to be able to successfully simulate 
$\Pi$. However, it follows that $N -  (\epsilon/\epsilon') N \geq n = N (1 - 
\epsilon')$ which simplifies to $1 - (\epsilon/\epsilon') \geq 1 - \epsilon'$, 
or $\epsilon'^2 \geq \epsilon$, implying that the communication rate of $1 - 
\epsilon'$ can be at most $1 - \Omega(\sqrt{\epsilon})$, where $\epsilon$ is the 
fraction of errors applied by the channel. 


One can note that a main reason for the $1-\Omega(\sqrt{\epsilon})$ limitation 
in the above argument is that the adversary can target the rounds with 
redundant information in the relevant window. For instance, in the interactive 
coding scheme of \cite{Haeupler14}, the rounds with control information are in 
predetermined positions of the encoded protocol, and so, the adversary knows 
exactly which locations to corrupt.

Our idea for overcoming the aforementioned limitations in the case of an 
\emph{oblivious} adversarial channel is to employ some type of 
\textbf{information hiding} to hide the locations of the redundant rounds carrying control/verification 
information. In particular, we randomize the locations of control information 
bits within the output protocol, which allows us to guard against attacks that 
target solely the redundant information. In order to allow for this synchronized 
randomization in the standard \emph{private randomness} model assumed in this 
paper, Alice and Bob use the standard trick of first running an error-corrected 
randomness exchange procedure that allows them to establish some shared 
randomness hidden from the oblivious adversary that can be used for the 
rest of the simulation. Note that this inherently does not work for a 
\emph{fully adaptive} adversary, as the adversary can adaptively choose which 
locations to corrupt based on any randomness that has been shared over the 
channel. In fact, we believe that beating the $1-\Omega(\sqrt{\epsilon})$ 
communication rate barrier against fully adaptive adversaries may be 
fundamentally impossible for precisely this reason. 

Information hiding, while absolutely crucial, does not, however, make use of a 
larger average message length which, according to the conjectures of 
\cite{Haeupler14}, is necessary to beat the $1-\Omega(\sqrt{\epsilon})$ barrier. 
The idea we use for this, as already demonstrated in Section \ref{sec:random}, 
is 
the use of blocking and the subsequent application of error-correcting codes on 
each such block. 

Unfortunately, the same argument as given above shows that a 
straightforward application of \emph{block} error-correcting codes, as done in 
Section \ref{sec:random}, cannot work against an oblivious adversarial channel. 
The reason is that in such a case, an application of \emph{systematic} block 
error-correcting codes would be possible as well, and such codes again have 
pre-specified positions of redundancy which can be targeted by the adversarial 
channel. In particular, one could again disable all redundant rounds including 
the non-systematic parts of block error-correcting codes in a large 
window of $(\epsilon/\epsilon')N$ rounds and make the remaining communication 
useless with few extra errors. More concretely, suppose that one simply encodes 
all blocks of data with 
a standard block error-correcting code. For such block codes, one needs to 
specify \emph{a priori} how much redundancy should be added,
and the natural direction would be to set the relative distance to, say, $100 
\epsilon$ given that one wants to prepare against an error rate of $\epsilon$. 
However, this would allow the adversary to corrupt a constant fraction (e.g., 
$1/200$) of error correcting codes beyond their distance, thus making a constant 
fraction of the communicated information essentially useless. This would lead to 
a communication rate of $1 - \Theta(1)$. It can again be easily seen that in 
this tradeoff, the best fixed relative distance one can choose for block 
error-correcting codes is essentially $\sqrt{\epsilon}$, which would 
lead to a rate loss of $H(\sqrt{\epsilon})$ for the error-correcting codes but 
would also allow the adversary to corrupt at most a $\sqrt{\epsilon}$ fraction 
of all 
codewords. This would again lead to an overall communication rate of $1 - 
\tilde{\Omega}(\sqrt{\epsilon})$.

Our solution to the hurdle of having to 
commit to a fixed amount of redundancy in advance is to use \textbf{rateless 
error-correcting codes}. Unlike block error-correcting codes with fixed 
block length and minimum 
distance, rateless codes encode a message into a potentially \emph{infinite} 
stream of symbols such that having access to enough uncorrupted symbols allows a 
party to decode the desired message with a resulting communication rate that 
\emph{adapts} to the true error rate without requiring \emph{a priori} 
knowledge of the error rate. Since it is not possible for Alice and Bob to know 
in advance which data bits the adversary will corrupt, rateless codes allow 
them to adaptively adjust the amount of redundancy for each communicated block, 
thereby allowing the correction of errors without incurring too 
great a loss in the overall communication rate.

\section{Main Result: Interactive Coding for Oblivious Adversarial Errors}

In this section, we develop our main result. We remind the reader that in the oblivious adversarial setting assumed throughout the rest of this paper, the adversary is allowed to corrupt up to an $\epsilon$ fraction of the total number of bits exchanged by Alice and Bob. The adversary commits to the locations of these bits before the start of the protocol. Alice and Bob will use randomness in their encoding, and one asks for a coding scheme that allows Alice and Bob to recover the transcript of the original protocol with exponentially high probability in the length of the protocol (over the randomness that Alice and Bob use) for any fixed error pattern chosen by the adversary.


For simplicity in exposition, we assume that the input protocol is 
\emph{binary}, so that the simulating output protocol will also be binary. 
However, the results hold virtually as-is for protocols over larger alphabet. 
We first provide a high-level overview of our construction of an encoded 
protocol. The pseudocode of the algorithm appears in Figure~\ref{fig:oblivious}.

\subsection{High-Level Description of Coding Scheme}\label{subsec:highlevel}
Let us describe the basic structure of our interactive coding scheme. Suppose 
$\Pi$ is an $n$-round binary input protocol with average message length $\ell 
\geq 
\mathrm{poly}(1/\epsilon)$. Using Lemma~\ref{lem:blocking}, we first produce a 
$B$-blocked binary protocol $\piblk$ with $n'$ rounds that simulates $\Pi$.

Our encoded protocol $\pienc$ will begin by having Alice and Bob performing a 
\emph{randomness exchange procedure}. More specifically, Alice will generate 
some number of bits from her private randomness and encode the random string 
using an error-correcting code of an appropriate rate and distance. Alice will 
then transmit the encoding to Bob, who can decode the received string. This 
allows Alice and Bob to maintain \emph{shared random bits}. The randomness 
exchange procedure is described in further detail in 
Section~\ref{subsec:randexch}.

Next, $\pienc$ will simulate the $B$-sized blocks (which we call 
\emph{$B$-blocks}) of $\piblk$ in order in a structured manner. Each $B$-block 
will be encoded as a string of $2B$ bits using a \emph{rateless code}, and the 
encoded string will be divided into \emph{chunks} of size $b < B$. For a 
detailed discussion on the encoding procedure via rateless codes, see 
Section~\ref{subsec:datasend}.

Now, $\pienc$ will consist 
of a series of $\Niter$ \emph{iterations}. Each iteration consists of 
transmitting $b'$ 
rounds, and we call such a $b'$-sized unit a \emph{mini-block}, where $b'>b$. 
Each mini-block will consist of $b$ 
\emph{data bits}, as well as $b'-b$ bits of \emph{control information}. The data 
bits in successive mini-blocks will taken from the successive $b$-sized chunks 
obtained by the encoding under the rateless code. Meanwhile, the control 
information bits are sent by Alice and Bob in order to check whether they are in 
sync with each other and to allow a \emph{backtracking} mechanism to tack place 
if they are not.

For a particular $B$-block that is being simulated, mini-blocks keep 
getting sent until the receiving party of the $B$-block is able to decode 
the correct $B$-block, after which Alice and Bob move on to the next $B$-block 
in $\Pi$.

In addition to data bits, each mini-block also contains $b'-b$ bits of control 
information. A party's unencoded control information during a mini-block 
consists of some hashes of his view of the current state of the protocol as well 
as some backtracking parameters. The aforementioned quantities are encoded 
using a hash for verification as well as an error-correcting code. Each party 
sends his encoded control information as part of each mini-block. The locations 
of the control information within each mini-block will be randomized for 
the sake of \emph{information hiding}, using bits from the shared randomness of 
Alice and Bob. This is described in further detail in 
Section~\ref{subsec:controlinfo}. Moreover, we note that the hashes
used for the control information in each mini-block are seeded using bits from 
the shared randomness. The structure of each mini-block is shown in 
Figure~\ref{fig:protocolstructure}.

After each iteration, Alice and Bob try to decode each other's control 
information in order to determine whether they are in sync. If not, the parties 
decide whether to backtrack in a controlled manner (see 
Section~\ref{subsec:backtrack} for details).

Throughout the protocol, Alice maintains a \emph{block index} $c_A$ (which 
indicates which block of $\piblk$ she believes is currently being simulated), a 
\emph{chunk counter} $j_A$, a \emph{transcript} (of the blocks in 
$\piblk$ that have been simulated so far) $T_A$, a \emph{global counter} $m$ 
(indicating the number of the current iteration), a \emph{backtracking 
parameter} $k_A$, as well as a \emph{sync parameter} $\synca$. Similarly, Bob 
maintains $c_B$, $j_B$, $T_B$, $m$, $k_B$, and $\syncb$.

\begin{figure}[h]
 \centering
 \includegraphics[width=5in]{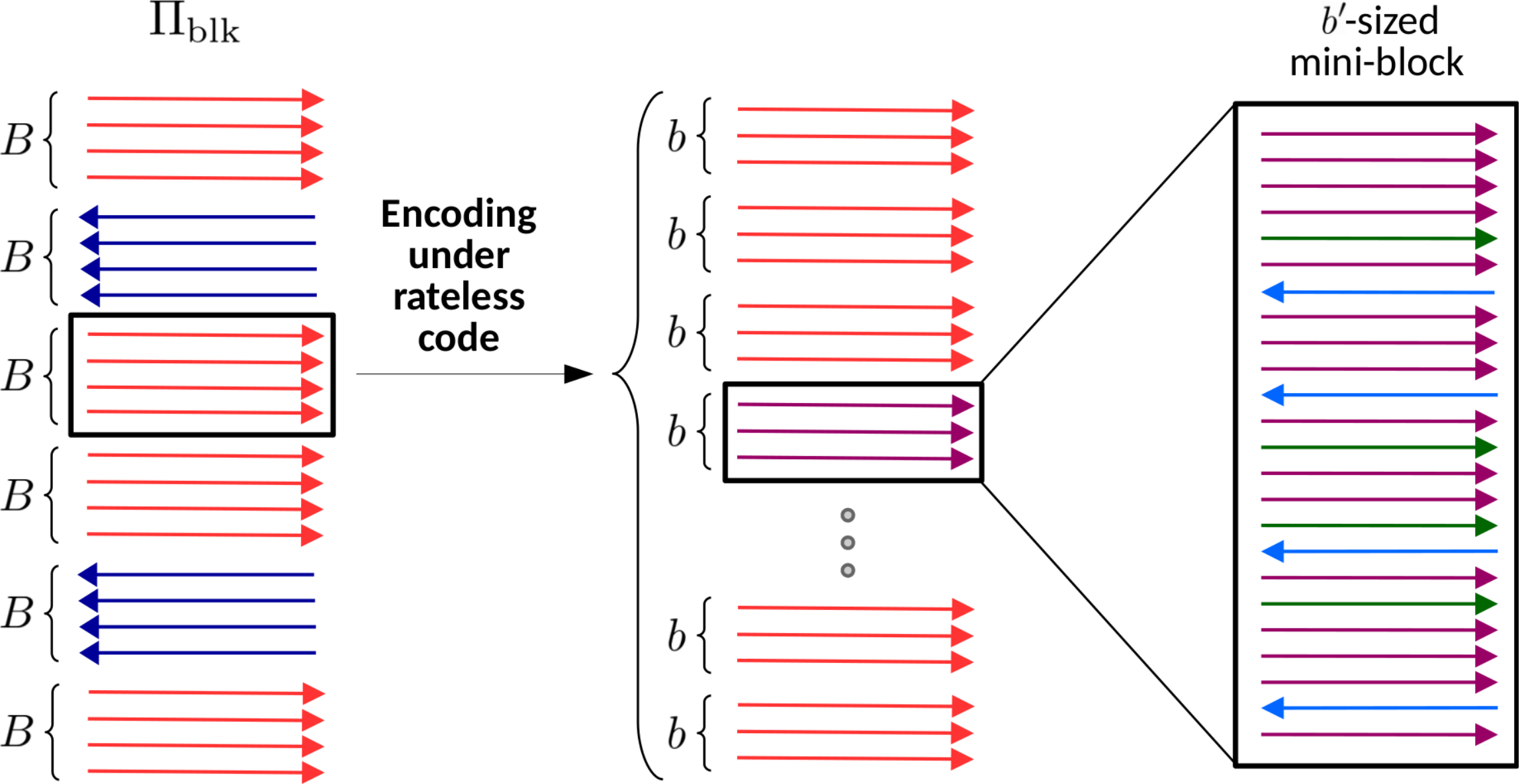}
 \caption{Each $B$-block of $\piblk$ gets encoded into 
chunks of size $b$ using a rateless code. Every $b'$-sized mini-block 
in $\pienc$ consists of the $b$ bits of such a chunk, along with $(b'-b)/2$ 
bits of Alice's control information and $(b'-b)/2$ bits of Bob's control 
information. The positions of the control information within a mini-block 
are randomized. Note that rounds with Alice's control information are in 
green, while rounds with Bob's control information are in light blue.}
 \label{fig:protocolstructure}
\end{figure}

\subsection{Parameters}
We now set the parameters of the protocol. For convenience, we will define a 
\emph{loss parameter} $\epsilon' < \epsilon$. Our interactive coding scheme 
will incur a rate loss of $\Theta(\epsilon'\,\mathrm{polylog}(1/\epsilon'))$, 
in 
addition to the usual rate loss of $\Theta(H(\epsilon))$. Alice and Bob are free 
to decide on an $\epsilon'$ based on what rate loss they are willing to tolerate 
in the 
interactive coding scheme. In particular, note that if $\epsilon' = 
\Theta(\epsilon^2)$, then the rate loss of 
$\Theta(\epsilon'\,\mathrm{polylog}(1/\epsilon'))$ is overwhelmed by 
$\Theta(H(\epsilon))$. For the purposes of Theorem~\ref{thm:mainoblivious}, it 
will suffice to take $\epsilon' = \Theta(\epsilon^2)$ at then end, but for the 
sake of generality, we maintain $\epsilon'$ as a separate parameter.

We now take the average message length threshold to be $\Omega(1/\epsilon'^3)$, 
i.e., we assume that our input protocol $\Pi$ has average message length $\ell 
= \Omega(1/\epsilon'^3)$. Then, $\Pi$ has at most $\nalt = n/\ell = 
O(n\epsilon'^3)$ alternations. Moreover, we take $B = 
\Theta(1/\epsilon'^2)$ and $b = s = \Theta(1/\epsilon')$, with $B = sb$. Then, 
by Lemma~\ref{lem:blocking}, note that $n' \leq n + \nalt\cdot B = n(1 + 
O(\epsilon'))$.

We also take $b' = b +
2c\log(1/\epsilon')$, so that within each $b'$-sized mini-block, each party 
transmits $c\log(1/\epsilon')$ bits of (encoded) control information.

Finally, we take $\Niter = \frac{n'}{b}(1 + 
\Theta(\epsilon\log(1/\epsilon))$ iterations. This will guarantee, with high 
probability, that at the end of the protocol, Alice and Bob have successfully 
simulated all blocks of $\piblk$, and therefore, $\Pi$. Also, it should be 
noted that we append trivial blocks of zeros (sent by, say, Alice) to the end 
of $\piblk$ to simulate in case $\pienc$ ever runs out of blocks of $\piblk$ to 
simulate (because it has reached the bottom of the protocol tree) before 
$\Niter$ iterations of $\pienc$ have been executed.

\subsection{Randomness Exchange}\label{subsec:randexch}
Alice and Bob will need to have some number of shared random bits throughout
the course of the protocol. The random bits will be used for two main purposes:
\emph{information hiding} and \emph{seeding hash functions}, which will be
discussed in Section~\ref{subsec:controlinfo}. As it turns out, it will suffice 
for Alice and Bob to have $l' = 
O(n\epsilon'\,\mathrm{polylog}(1/\epsilon'))$ shared random bits for the 
entirety of the protocol, using some additional tricks.

Thus, in the private randomness model, it suffices for Alice to generate the 
necessary number of random bits and transmit them to Bob using an 
error-correcting code. More precisely, Alice generates a uniformly random 
string $\shrand\in\{0,1\}^{l'}$, uses an error-correcting code $\cexch: 
\{0,1\}^{l'} \to \{0,1\}^{10\epsilon\Niter b'}$ of relative distance $2/5$ to 
encode $\shrand$, and transmits the encoded string to Bob. Since the adversary 
can corrupt only at most $\epsilon$ fraction of all
bits, the transmitted string cannot be corrupted beyond half the minimum
distance of $\cexch$. Hence, Bob can decode the received string and determine 
$\shrand$.

Note that the exchange of randomness via the codeword in $\cexch$ results in a 
rate loss of $\Theta(\epsilon)$, which is still overwhelmed by 
$\Theta(H(\epsilon))$.

\subsection{Sending Data Bits Using ``Rateless'' Error-Correcting
Codes}\label{subsec:datasend}
To transmit data from blocks of $\piblk$, we will use an error-correcting code 
that has incremental distance properties. One can think of this as a rateless 
code with minimum distance properties. Recall 
that $b = s = \Theta(1/{\epsilon'})$ and $B=sb$. In particular, we require an
error-correcting code $\crateless: \{0,1\}^{B} \to \{0,1\}^{2B}$ for which
the output is divided in to $2s$ chunks of $b$ bits each such that the code
restricted to any contiguous block (with cyclic wrap-around) of $> s$ chunks
has a certain guaranteed minimum distance. The following lemma guarantees the
existence of such a code.
\shortOnly{The proof appears in Appendix~\ref{app:proofs}.}

\begin{lemma}\label{lem:ecc}
 For sufficiently large $b, s$, there exists an error-correcting code $\mcC: 
\{0,1\}^{sb} \to \{0,1\}^{2sb}$  such 
that for any $a = 0,1,\dots, 2s-1$ and $j = s+1, s+2,\dots, 2s$, the code 
$\mcC_{a,j}:\{0,1\}^{sb} \to \{0,1\}^{jb}$ formed by 
restricting $\mcC$ to the bits $ab, ab+1, \dots, ab+jb-1$ (modulo $2sb$) has 
relative distance at least $\delta_{j} =
H^{-1}\left(\frac{j-s}{j}
- \frac{1}{4s}\right)$, while $\mcC$ has relative distance at least
$\delta_{2s} = \frac{1}{15}$. (Here, $H^{-1}$ denotes the unique inverse of $H$
that takes
values in $[0,1/2]$.)
\end{lemma}

\global\def\ProofofRatelessCode{
\begin{proof}[Proof of Lemma~\ref{lem:ecc}]
 We use a slight modification of the random coding argument that is often used
to 
establish the Gilbert-Varshamov bound. Suppose we pick a random 
linear code. For $s < j \leq 2s$, let us consider the probability $P_{a,j}$ that
the 
resulting $\mcC_{a,j}$ 
does not have relative distance at least $\delta_j$. Consider any codeword 
$y\in\{0,1\}^{jb}$ in $\mcC_{a,j}$. The probability that $y$ has Hamming weight
less 
than $\delta_j$ is at most $2^{-jb(1-H(\delta_j))}$. Thus, by the 
union bound, we have that the probability that $\mcC_{a,j}$ contains a codeword
of 
Hamming weight less than $\delta_j$ is at most
\begin{align*}
 P_{a,j} = 2^{sb} \cdot 2^{-jb(1-H(\delta_j))} &= 2^{sb-jb\left(1 -
\frac{j-s}{j} + \frac{1}{4s}\right)}\\
 &= 2^{-jb / 4s}\\
 &\leq 2^{-b/4}.
\end{align*}
Similarly, $P$, the probability that $\mcC$ contains a codeword of Hamming 
weight less than $\frac{2}{15}s$, is at most
\begin{align*}
 P \leq 2^{sb} \cdot 2^{-2sb(1-H(2/15))} \leq 2^{-sb/4} \leq 2^{-b/4}.
\end{align*}

Therefore, by another application of the union bound, the probability that some 
$\mcC_{a,j}$ or $\mcC$ does not have the required relative distance is at most
\begin{align*}
 P + \sum_{\substack{0\leq a \leq 2s-1\\ s <  j\leq 2s}} P_{a,j} \leq (2s^2+1)
\cdot
2^{-b/4} < 1
\end{align*}
for sufficiently large $b, s$.
\end{proof}
\begin{remark}
For our purposes, $b = s = \Theta(1/\epsilon')$. Therefore, for suitably small 
$\epsilon' > 0$, there exists such an error-correcting code $\mcC$ as 
guaranteed by Lemma~\ref{lem:ecc}. Moreover, it is possible to find a such a 
code by brute force in time $\mathrm{poly}(1/\epsilon')$.
\end{remark}
}

\fullOnly{\ProofofRatelessCode}

Thus, Alice and Bob can agree on a fixed error-correcting code $\crateless$
of the type guaranteed by Lemma~\ref{lem:ecc} prior to the start of the 
algorithm. Now, let us describe how data bits are sent during the iterations of 
$\pienc$. The blocks of $\pienc$ are 
simulated in order as follows.

First, suppose Alice's block index $c_A$ indicates a 
$B$-block in $\piblk$ during which Alice is the sender. Then in $\pienc$,  
Alice will transmit 
up to a maximum of $2s$ chunks (of size 
$b$) that 
will encode the data $x$ from 
that block. More specifically, Alice will compute $y = \crateless(x) 
\in \{0,1\}^{2B}$ and decompose it as $y = y_0 \circ y_1 \circ \cdots \circ 
y_{2s-1}$, where $\circ$ denotes concatenation 
and $y_0, y_1, \dots, y_{2s-1} \in \{0,1\}^b$.

Recall that each mini-block of 
$\pienc$ contains $b$ data 
bits (in 
addition to $b'-b$ control bits). Thus, Alice can send each $y_i$ as the 
data 
bits of a mini-block. The chunk that Alice sends in a given iteration depends
on the global counter $m$. In particular, Alice always sends the chunk
$y_{m\bmod 2s}$. Moreover, Alice keeps a chunk counter $j_A$, which is set to 0 
during the first iteration in which she transmits a chunk from $y$ and then 
increases by 1 during each subsequent iteration (until $j_A = 2s$, at which 
point $j_A$ stops increasing).

On the other hand, suppose Alice's block index $c_A$ indicates a $B$-block in 
$\piblk$ during which Alice is the receiver. Then, Alice listens for data 
during each mini-block. Alice stores her received $b$-sized chunks as 
$\widt{g}_0, \widt{g}_1, \dots$ and increments her chunk counter $j_A$ after 
each iteration to keep track of how many chunks she has stored, along with $a$, 
an index indicating which $y_a$ she expects the first chunk $\widt{g}_0$ to be. 
Once Alice has received more than $s$ chunks (i.e., $j_A > s$), she starts to 
keep an estimate $\widt{x}$ of the data $x$ that Bob is sending
that Alice has by decoding $\widt{g}_0 \circ \widt{g}_1 \circ \cdots \circ 
\widt{g}_{j_A-1}$ to the nearest codeword of $\crateless_{a,j_A}$. This 
estimate is updated after each subsequent iteration. As soon as Alice undergoes 
an iteration in which she receives valid control information 
suggesting that $\widt{x}=x$ (if Alice's estimate $\widt{x}$ matches the hash 
of $x$ that Bob sends as control 
information, see Section~\ref{subsec:controlinfo}), she advances her block 
index $c_A$ and appends her transcript $T_A$ with $\widt{x}$.

Note that it is possible that $j_A$ reaches $2s$ and Alice has not yet received 
valid control information suggesting that he has decoded $x$. In this case, 
Alice resets $j_A$ to 0 and 
also resets $a$ to the current value of $m$, thereby restarting the listening 
process. Also, during any iteration, if Alice receives control information 
suggesting that $j_B < j_A$ (i.e., Alice has been listening for a greater 
number of iterations than Bob has been transmitting), then again, Alice resets 
$j_A$ and $a$ and restarts the process.

\begin{remark}
The key observation is that using a rateless code
allows the amount of redundancy in data that the sender sends to \emph{adapt} to
the number of errors being introduced by the adversary, rather than wasting
redundant bits or not sending enough of them.
\end{remark}

\subsection{Control Information}\label{subsec:controlinfo}
Alice's unencoded control information in the $m^\text{th}$ iteration consists of
(1.) a hash 
$h_{A,c}^{(m)} =
{hash}(c_A, S)$ of the block index $c_A$,
(2.) a hash $h_{A,x}^{(m)} = {hash}(x, S)$ of the data in the current 
block 
of $\piblk$ being communicated, (3.) a hash $h_{A, k}^{(m)} = 
{hash}(k_A, S)$ 
of the
\emph{backtracking
parameter} $k_A$, (4.) a hash $h_{A,T}^{(m)} = {hash}(T_A, S)$ of Alice's
transcript $T_A$, (5.) a hash $h_{A, {\tt MP1}}^{(m)} = {hash}(T_A[1,{\tt 
MP1}], 
S)$ 
of
Alice's transcript up till the first \emph{meeting point}, (6.) a hash
$h_{A,{\tt MP2}}^{(m)} = {hash}(T_A[1,{\tt MP2}], S)$ of Alice's transcript up 
till the
second \emph{meeting point}, (7.) the chunk counter $j_A$, and (8.) the 
\emph{sync parameter} $\synca$. Here, $S$ refers 
to a string of fresh random bits used to seed the hash functions (note that $S$ 
is different for each instance). 
Thus, we write Alice's unencoded control information as
\[
 \ctrlinfo_A^{(m)} = \left(h_{A,c}^{(m)}, h_{A,x}^{(m)}, h_{A,k}^{(m)}, 
h_{A,T}^{(m)},
h_{A, {\tt MP1}}^{(m)}, h_{A, {\tt MP2}}^{(m)}, j_A, \synca \right).
\]
Bob's unencoded control information $\ctrlinfo_B^{(m)}$ is similar in the 
analogous way.

For the individual hashes, we can use the following Inner Product hash function 
${hash}: \{0,1\}^l \times \{0,1\}^r \to \{0,1\}^{p}$, where $r=lp$:
\[
 {hash}(X, R) = \left(\langle X, R_{[1,l]}\rangle, \langle X, 
R_{[l+1,2l]}\rangle, \dots, \langle X, R_{[lp-(l-1), lp]}\rangle\right),
\]
where the first argument $X$ is the quantity to be hashed, and the second 
argument $R$ is a random seed. This choice of hash function guarantees the 
following property:
\begin{property}\label{prop:collision}
 For any $X, Y\in \{0,1\}^l$ such that $X\neq Y$, we have that 
$\Pr_{R\sim\mathrm{Unif}(\{0,1\}^r)} [{hash}(X, R) = {hash}(Y, R)] \leq 2^{-p}$.
\end{property}
Now, we wish to take output size $p=O(\log(1/\epsilon'))$ for each of the 
hashes so that the total size of each party's control information in any 
iteration is $O(\log(1/\epsilon'))$. Note that some of the quantities we 
hash (e.g., $T_A$, $T_B$) actually have size $l = \Omega(n)$. Thus, for the 
corresponding hash function, we would naively require $r = lp = 
\Omega(n\log(1/\epsilon'))$ fresh bits of randomness for the seed (per 
iteration), for a total of $\Omega(\Niter n\log(1/\epsilon'))$ bits 
of randomness. However, as described in Section~\ref{subsec:randexch}, Alice 
and Bob only have access to $O(n\epsilon' \mathrm{polylog}(1/\epsilon'))$ bits 
of shared randomness!

To get around this problem, we make use of $\delta$-biased sources to minimize 
the amount of randomness we need. In particular, we can use the $\delta$-biased 
sample space of \cite{NaorNaor} to stretch $\Theta(\log(L/\delta))$ independent 
random bits into a string of $L = \Theta(\Niter 
n\log(1/\epsilon'))$ pseudorandom bits that are $\delta$-biased. We take 
$\delta = 2^{-\Theta(\Niter\cdot p)}$. The sample space guarantees that the $L$ 
pseudorandom bits are $\delta^{\Theta(1)}$-statistically close to being 
$k$-wise independent for $k=\log(1/\delta) = \Theta(\Niter\cdot p) = 
\Theta(\Niter\log(1/\epsilon'))$. Moreover, the Inner Product Hash Function 
satisfies the following modified collision property, which follows trivially 
from Property~\ref{prop:collision} and the definition of $\delta$-bias:
\begin{property}\label{prop:collisionbiased}
 For any $X, Y\in \{0,1\}^l$ such that $X\neq Y$, we have that 
$\Pr_{R} [{hash}(X, R) = {hash}(Y, R)] \leq 2^{-p} + \delta$,
where $R$ is sampled from a $\delta$-biased source.
\end{property}
As it turns out, this property is good enough for our purposes. Thus, 
after the randomness exchange, Alice and Bob can simply take 
$\Theta(\log(L/\delta))$ bits from $\shrand$ and stretch them into an $L$-bit 
string $\shpseudo$ as described. Then, for each iteration, Alice and Bob can 
simply seed their hash functions using bits from $\shpseudo$.

\subsubsection{Encoding and Decoding Control
Information}\label{subsubsec:controlencode}
Recall that during the $m^\text{th}$ iteration, Alice's (unencoded) control 
information is $\ctrlinfo_A^{(m)}$, while Bob's (unencoded) control information 
is 
$\ctrlinfo_B^{(m)}$. In this section, we describe the encoding and 
decoding functions that 
Alice and Bob use for their control information. We start by listing the 
properties we desire.

\begin{definition}
 Suppose $X\in\{0,1\}^l$ and $V\in\{*,\neg,0,1\}^l$ for some $l > 0$. Then, we 
define $\corrupt_V(X) = Y \in \{0,1\}^l$ as follows:
\[
 Y_i = \begin{cases}
        V_i \quad &\text{if $V_i \in \{0,1\}$}\\
        X_i \oplus 1 \quad &\text{if $V_i = \neg$}\\
        X_i \quad &\text{if $V_i = *$}
       \end{cases}.
\]
Moreover, we define $\cwt(V)$ to be the number of coordinates of $V$ that are 
not equal to $*$.
\end{definition}
\begin{remark}
 Note that $V$ corresponds to an error pattern. In 
particular, $*$ indicates a position that is not corrupted, while $\neg$ 
indicates a bit flip, and 0/1 indicate a bit that is fixed to the appropriate 
symbol (see Section~\ref{subsec:commchannel} for details about \emph{flip} and 
\emph{replace} errors). The function $\corrupt_V$ applies the error pattern $V$ 
to the bit 
string given as an argument. Also, $\wt(V)$ corresponds to the number of 
positions that are targeted for corruption.
\end{remark}

We require a seeded encoding function $\contenc: 
\{0,1\}^{l} \times \{0,1\}^r \to 
\{0,1\}^{o}$ as well as a seeded decoding function $\contdec: 
\{0,1\}^{o} \times \{0,1\}^r \to 
\{0,1\}^{l} \cup \{\perp\}$ such that the following property holds:
\begin{property}\label{prop:encdec}
The following holds:
\begin{enumerate}
 \item For any $X\in\{0,1\}^{l}$, $R\in\{0,1\}^{r}$, and $V\in\{*, 
\neg, 0, 1\}^{o}$ such 
that $\cwt(V) < \frac{1}{8}o$,
\[
 \contdec(\corrupt_V(\contenc(X, R)), R) = X.
\]
\item For any $X \in\{0,1\}^{l}$ and $V\in\{0,1\}^{o}$ such 
that $\cwt(V) \geq \frac{1}{8}o$,
\[
 \Pr_{R\sim\Unif(\{0,1\}^r)} \left[\contdec(\corrupt_V(\contenc(X,R)), 
R) 
\not\in \{X, \perp\} \right] \leq 2^{-\Omega(l)}.
\]
\end{enumerate}
\end{property}
\begin{remark}
The second argument of $\contenc$ and $\contdec$ will be a 
\emph{seed}, which is generated by taking $r$ fresh bits from the shared 
randomness of Alice and Bob. A decoding output of $\perp$ indicates a 
decoding failure. Moreover, 
(1.) of Property~\ref{prop:encdec} guarantees that a party can successfully 
decode the other party's control information if at most a constant fraction of 
the encoded control information symbols are corrupted (this is then used to 
prove Lemmas~\ref{lem:whpcontrol} and \ref{lem:invbound}). On the other 
hand, (2.) of Property~\ref{prop:encdec} guarantees that if a larger fraction 
of the encoded control information symbols are corrupted, then the decoding 
party can detect any possible corruption with high probability (this is 
then used to establish Lemma~\ref{lem:malbound}).
\end{remark}

\global\def\EncodingScheme{
We now show how to obtain $\contenc$, $\contdec$ that satisfy 
Property~\ref{prop:encdec}. The idea is that $\contenc$ consists of a 
three-stage encoding: (1.) append a hash value to the unencoded control 
information, (2.) encode the resulting string using an error-correcting code, 
and (3.) XOR each output bit with a fresh random bit taken from the shared 
randomness.

For our purposes, we want $l = O(\log(1/\epsilon'))$ to be the number of bits 
in $\ctrlinfo_A^{(m)}$ (or $\ctrlinfo_B^{(m)}$) and $o = c\log(1/\epsilon')$.

First, we choose a hash function $h: \{0,1\}^l \times \{0,1\}^t \to 
\{0,1\}^{o'}$ that has the following property:
\begin{property}\label{prop:addshift}
 Suppose $X, U \in \{0,1\}^l$, where $U$ is not the all-zeros vector, and 
$W\in\{0,1\}^{o'}$. Then,
 \[
  \Pr_{R\sim\mathrm{Unif}(\{0,1\}^t)} [h(X+U, R) = h(X,R) + W] \leq 2^{-o'}.
 \]
\end{property}
\noindent In particular, we can use the simple Inner Product Hash 
Function with $t 
= l\cdot o'$ and $o' = \Theta(\log(1/\epsilon'))$:
\[
 h(X,R) = \left(\left\langle X, R_{[1,l]}\right\rangle, \left\langle X, 
R_{[l+1,2l]}\right\rangle, \dots, \left\langle X, R_{[l\cdot o' - (l-1), l\cdot 
o']}\right\rangle\right).
\]

Next, we choose a \emph{linear} error-correcting code $\chash:\{0,1\}^{l+o'} 
\to \{0,1\}^{o}$ of constant 
relative distance $1/4$ and constant rate.

We now take $r = t + o$ and define $\contenc$ as
\[
 \contenc(X,R) = \chash(X \circ h(X, R_{[o+1,r]})) \oplus R_{[1,o]}.
\]
Moreover, we define $\contdec$ as follows: Given $Y, R$, let $X'$ be the 
decoding of $Y + R_{[1,o]}$ under $\chash$ (using the nearest 
codeword of $\chash$ and then inverting the map $\chash$). We then define
\[
 \contdec(Y,R) = \begin{cases}
                  X'_{[1,l]} \quad &\text{if 
$h(X'_{[1,l]}, R_{[o+1,r]}) = X'_{[l+1,l+o']}$}\\
                  \perp \quad &\text{if $h(X'_{[1,l]}, R_{[o+1,r]}) 
\neq X'_{[l+1,l+o']}$}
                 \end{cases}.
\]
\begin{remark}
 Note that we have $r = O(\log^2(1/\epsilon'))$, which means that over the 
course of the protocol $\pienc$, we will need $O(\Niter r) = 
O(n\epsilon'\log^2(1/\epsilon'))$ fresh random bits for the purpose of encoding 
and decoding control information.
\end{remark}

We now prove that the above $\contenc$, $\contdec$ satisfy 
Property~\ref{prop:encdec}.

\begin{proof}
 Note that if $V\in\{*,\neg,0,1\}^o$ satisfies $\cwt(V) < \frac{1}{8}o$, then 
note that the Hamming distance between $\corrupt_V(\contenc(X,R))$ and 
$\contenc(X,R)$ is less than $\frac{1}{8}o$. Hence, since $\chash$ has relative 
distance $1/4$, it follows that under the 
error-correcting code $\chash$, $\corrupt_V(\contenc(X,R)) \oplus 
R_{[1,o]}$ and $\contenc(X,R) \oplus R_{[1,o]}$ decode to 
the same element of $\{0,1\}^{l+o'}$, namely, $X\circ h(X,R)$. Part (1.) of 
Property~\ref{prop:encdec} therefore holds.

Now, let us establish (2.) of Property~\ref{prop:encdec}. Consider a 
$V\in\{0,1\}^o$ with $\cwt(V) \geq \frac{1}{8}o$. Now, let us enumerate 
$W^{(1)}, W^{(2)}, \dots, W^{(2^{\cwt(V)})} \in\{0,1\}^o$ as the set of all 
$2^{\cwt(V)}$ vectors in $\{0,1\}^o$ which have a 0 in all coordinates where 
$V$ has a $*$. Now, observe that the distribution of $\corrupt_V(\contenc(X, 
R))$ over $R_1, R_2, \dots, R_o$ taken i.i.d. uniformly in $\{0,1\}$ is 
identical to the distribution of
\[
 \chash(X\circ h(X, R_{[o+1,r]})) \oplus W,
\]
where $W$ is chosen uniformly from $\left\{W^{(1)}, W^{(2)}, \dots, 
W^{(2^{\cwt(V)})}\right\}$. Now, note that for each $W^{(i)}$, there exists a 
corresponding $U^{(i)} \in \{0,1\}^{o+l}$ such that under the nearest-codeword 
decoding of $\chash$,
\[
\chash(X\circ h(X, R_{[o+1, r]}))) 
\oplus W^{(i)}
\]
decodes to $(X\circ h(X, R_{[o+1, r]})) \oplus U^{(i)}$. 
Thus, we have that
\begin{multline*}
 \Pr_{R\sim\Unif(\{0,1\}^r)}\left[\contdec(\corrupt_V(\contenc(X,R)),R)\not\in 
\{X, \perp\}\right]\\ = \Pr_{\substack{R_{o+1} \dots, R_r \sim \Unif(\{0,1\})\\ 
1\leq i \leq 2^{\cwt(V)}}} \left[U^{(i)} \neq (0,0,\dots,0) \text{ AND }   
h\left(X\oplus U^{(i)}_{[1,l]}\right) = h\left(X,R_{[o+1,r]}\right) \oplus 
U^{(i)}_{[l+1,l+o]} \right],
\end{multline*}
which, by Property~\ref{prop:addshift}, is at most $2^{-o'}$, thereby 
establishing (2.) of Property~\ref{prop:encdec}.
\end{proof}
}

\fullOnly{\EncodingScheme}
\shortOnly{
For our purposes, we take $l = O(\log(1/\epsilon'))$ to be the number 
of bits in $\ctrlinfo_A^{(m)}$ (or $\ctrlinfo_B^{(m)}$) as well as $o = 
c\log(1/\epsilon')$ and $r = \Theta(\log^2 (1/\epsilon'))$. Thus, over the 
course of the protocol $\pienc$, we will need $O(\Niter r) = 
O(n\epsilon'\log^2(1/\epsilon'))$ fresh random bits for the purpose of encoding 
and decoding control information. The details of the construction for 
$\contenc$, $\contdec$ are provided in Appendix~\ref{app:encscheme}.
}

\subsubsection{Information Hiding}\label{subsubsec:infohide}
We now describe
how the encoded control information bits are sent within each mini-block. 
Recall that in the $m^\text{th}$ iteration, Alice 
chooses a fresh random seed $R^A$ taken from the 
shared randomness $\shrand$ and computes her encoded 
control information $\contenc(\ctrlinfo_A^{(m)}, R^A)$. Similarly, Bob chooses 
$R^B$ and computes $\contenc(\ctrlinfo_B^{(m)}, R^B)$. Recall that $R^A, 
R^B$ are known to both Alice and Bob.

As discussed previously, the control information bits in each mini-block are
not sent contiguously. Rather, the 
locations of the control information bits within each $b'$-sized mini-block are 
hidden from 
the oblivious adversary by using the shared randomness to agree on a designated 
set of $2c\log(1/\epsilon')$ locations. In particular, the locations of the 
control information bits sent by Alice and 
Bob during the $m^\text{th}$ iteration are given by the variables $z_{m,i}^A$
and $z_{m,i}^B$ ($i=1,\dots, c\log(1/\epsilon')$), respectively. For each 
$m$, these 
variables are chosen randomly at the beginning using $O(\log^2(1/\epsilon'))$ 
fresh random bits from the preshared string $\shrand$. Since there are $\Niter$ 
iterations, this will
require a total of $\Theta(\Niter \cdot
\log^2(1/\epsilon')) = \Theta(n\epsilon' \log^2(1/\epsilon'))$ random 
bits from $\shrand$.

Thus, Alice sends the $c\log(1/\epsilon')$ bits of $\contenc(\ctrlinfo_A^{(m)}, 
R^A)$ in positions $z_{m,i}^A$ ($i=1,\dots, c\log(1/\epsilon')$) of the 
mini-block of the $m^\text{th}$ iteration, and similarly, Bob sends the bits of 
$\contenc(\ctrlinfo_B^{(m)}, R^B)$ in positions $z_{m,i}^B$ ($i=1,\dots, 
c\log(1/\epsilon')$). Meanwhile, Bob listens for Alice's encoded control 
information in 
positions $z_{m,i}^A$ of the mini-block and assembles the received bits as a 
string $Y \in \{0,1\}^{c\log(1/\epsilon')}$, after which Bob tries to decode 
Alice's control information by computing $\contdec(Y, R^A)$. Similarly, Alice 
listens for Bob's encoded control information in locations $z_{m,i}^B$ and tries 
to decode the received bits.

After each iteration, Alice and Bob use their decodings of each other's 
control information to decide how to proceed. This is described in detail in 
Section~\ref{subsec:backtrack}.

\begin{remark}
 The information hiding provided by the randomization of $z_{m,i}^A$ 
and $z_{m,i}^B$ ($i=1,\dots,c\log(1/\epsilon')$) ensures that an oblivious
adversary generally needs to corrupt a constant fraction of bits in a
mini-block in order to corrupt a constant fraction of either party's encoded 
control
information bits in that mini-block. Along with Property~\ref{prop:encdec}, 
this statement is used to prove Lemma~\ref{lem:whpcontrol}.
\end{remark}

\subsection{Flow of the Protocol and Backtracking}\label{subsec:backtrack}
Throughout $\pienc$, each party maintains a state that indicates whether 
both parties are in sync as well as parameters that allow for backtracking 
in the case that the parties are not in sync. After each iteration, Alice and 
Bob use their decodings of the other party's control information from that 
iteration to update their states. We describe the flow of the protocol in 
detail.

Alice and Bob maintain binary variables $\synca$ and $\syncb$, respectively, 
which indicate the players' individual perceptions of whether they are in sync. 
Note that $\synca = 1$ implies $k_A = 1$ (and similarly, $\syncb=1$ implies 
$k_B=1$). Moreover, in the case that $\synca = 1$ (resp. $\syncb=1$), the 
variable 
$\spka$ (resp. $\spkb$) indicates whether Alice (resp. Bob) speaks in the 
$c_A^{\text{th}}$ (resp. $c_B^{\text{th}}$) block of $\piblk$, based on the 
transcript thus far.

Let us describe the protocol from Alice's point of view, as Bob's procedure is 
analogous. Note that after each iteration, Alice attempts to decode Bob's 
control information for that iteration. We say that Alice \emph{successfully 
decodes} Bob's control information if the decoding procedure (see 
Section~\ref{subsubsec:controlencode}) does not output $\perp$. In this case, 
we write the output of the control information decoder (for the $m^\text{th}$ 
iteration) as
\[
 \widt{\ctrlinfo}_B^{(m)} = \left(\widt{h}_{B,c}^{(m)}, \widt{h}_{B,x}^{(m)}, 
\widt{h}_{B,k}^{(m)}, \widt{h}_{B,T}^{(m)}, \widt{h}_{B, {\tt MP1}}^{(m)}, 
\widt{h}_{B, {\tt MP2}}^{(m)}, \widt{j}_B, \rcvsyncb\right).
\]
We now split into two cases, based on whether $\synca=1$ or $\synca=0$.

\bigskip
\noindent\underline{$\synca=1$:}

\bigskip
The general idea is that whenever Alice thinks she is in sync with
Bob (i.e., $\synca=1$), she either (a.) \emph{listens} for data bits from Bob 
while
updating her estimate $\widt{x}$ of block $c_A$ of $\piblk$, if
$\spka=0$, or (b.) \emph{transmits}, as data bits of the next iteration, the 
$(m\bmod{2s})$-th chunk of the encoding of $x$ (the $c_A$-th $B$-block of 
$\piblk$) under $\crateless$, if $\spka=1$ (see
Section~\ref{subsec:datasend} for details).

If Alice is \emph{listening} for data bits, then Alice expects that $k_A = k_B = 
1$ and either (1.) $c_A = c_B$, $T_A = T_B$ or (2.) $c_A = c_B + 1$, $T_B = 
T_A[1\dots (c_B-1)B]$. Condition (1.) is expected to hold if Alice has still 
not managed to decode the $B$-block $x$ that Bob is trying to relay, while (2.) 
is expected if Alice has managed to decode $x$ and has advanced her transcript 
but Bob has not yet realized this.

On the other hand, if Alice is \emph{transmitting} data bits, then Alice 
expects that $k_A = k_B = 1$, as well as either (1.) $c_A = c_B$, $T_A = T_B$, 
or (2.) $c_B = c_A + 1$, $T_B = T_A \circ x$, or (3.) $c_A = c_B+1$, $T_B = 
T_A[1\dots (c_B-1)B]$. Condition (1.) is expected to hold if Bob is still 
listening 
for data bits and has not yet decoded Alice's $x$, while (2.) is expected to 
hold if Bob has already managed to decode $x$ and advanced his block 
index and transcript, and (3.) is expected to hold if Bob has been transmitting 
data bits to Alice (for the $(c_A-1)$-th $B$-block of $\piblk$), but Bob has 
not realized that Alice has decoded the correct $B$-block and moved on.

Now, if Alice manages to successfully decode Bob's control information in the 
most recent iteration, then Alice checks whether the hashes 
$\widt{h}_{B,c}^{(m)}$, $\widt{h}_{B,k}^{(m)}$, $\widt{h}_{B,T}^{(m)}$, 
$\widt{h}_{B,x}^{(m)}$, as well as $\rcvsyncb$ are 
consistent with Alice's expectations (as outlined in the previous two 
paragraphs). If not, then Alice sets $\synca = 0$. Otherwise, Alice proceeds 
normally.

\begin{remark}
Note that in general, if a party is trying to transmit the contents $x$ of a 
$B$-block and the other party is trying to listen for $x$, then there is a 
delay of at least one iteration between the time that the listening party 
decodes $x$ and the time that the transmitting party receives control 
information suggesting that the other party has decoded $x$. However, since 
$b/B = O(\epsilon')$, the rate loss due to this delay turns out to be 
just $O(\epsilon')$.
\end{remark}
\bigskip
\noindent\underline{$\synca=0$:}

\bigskip
Now, we consider what happens when Alice believes she is out of sync (i.e.,
$\synca=0$). In this case, Alice uses a meeting point based backtracking
mechanism along the lines of \cite{Schulman92} and \cite{Haeupler14}. We sketch
the main ideas below:

Specifically, Alice keeps a backtracking parameter $k_A$ that is initialized as
1 when Alice first believes she has gone out of sync and increases by 1 each 
iteration
thereafter. (Note that $k_A$ is also maintained when $\synca=1$, but it is
always set to 1 in this case.) Alice also maintains a counter $E_A$ that counts
the number of discrepancies between $k_A$ and $k_B$, as well as \emph{meeting 
point counters} $v_1$ and $v_2$. The counters $E_A, v_1, v_2$ are initialized 
to zero when Alice first sets $\synca$ to 0.

The parameter $k_A$ measures the amount by which Alice is willing to backtrack
in her transcript $T_A$. More specifically, Alice creates a \emph{scale}
$\widt{k}_A = 2^{\lfloor \log_2
k_A \rfloor}$ by rounding $k_A$ to the largest power of two that does not
exceed it. Then, Alice defines two \emph{meeting points} ${\tt MP1}$ and ${\tt
MP2}$ on this scale to be the two largest multiples of $\widt{k}_A B$ not 
exceeding $|T_A|$. More precisely, ${\tt MP1} =
\widt{k}_A B \left\lfloor \frac{|T_A|}{k_A B} \right\rfloor$ and ${\tt MP2} =
{\tt MP1} - \widt{k}_A B$. Alice is willing to rewind her transcript to 
either one of $T_A[1\dots {\tt MP1}]$ and $T_A[1\dots {\tt MP2}]$, the last two 
positions in her transcript where the number of $B$-blocks of $\piblk$ that 
have been simulated is an integral multiple of $\widt{k}_A$.

If Alice is able to successfully decode Bob's control information, then she 
checks $\widt{h}_{B,k}^{(m)}$. If it does not agree with the hash of $k_A$ 
(suggesting that $k_A\neq k_B$), then Alice increments $E_A$. Alice also 
increments $E_A$ if $\rcvsyncb = 1$.

Otherwise, if $\widt{h}_{B,k}^{(m)}$ matches her computed hash of $k_A$, then 
Alice checks whether either of $\widt{h}_{B, {\tt MP1}}^{(m)}, \widt{h}_{B, 
{\tt 
MP2}}^{(m)}$ matches the appropriate hash of $T_A[1\dots {\tt MP1}]$. If so, 
then Alice 
increments her counter $v_1$, which counts the number of times her \emph{first} 
meeting point matches one of the meeting points of Bob. 
If not, then Alice then checks whether either of $\widt{h}_{B, {\tt 
MP1}}^{(m)}, \widt{h}_{B, {\tt MP2}}^{(m)}$ matches the hash of $T_A[1\dots 
{\tt 
MP2}]$ and if so, she increments her counter $v_2$, which counts the number of 
times her \emph{second} meeting point matches one of the meeting points of Bob.

In the case that Alice is not able to successfully decode 
Bob's control information from the most recent iteration (i.e., the decoder 
outputs $\perp$), she increments $E_A$.

Regardless of which of the above scenarios holds, Alice then increases $k_A$ by 
1 and updates $\widt{k}_A$, ${\tt MP1}$, and ${\tt MP2}$ accordingly.

Next, Alice checks whether to initiate a \emph{transition}. Alice only 
considers making a transition if $k_A = \widt{k}_A \geq 2$ (i.e., $k_A$ is a 
power of two and is $\geq 2$). Alice first decides whether to initiate a 
\emph{meeting point transition}. If $v_1 \geq 0.2 k_A$, then Alice rewinds 
$T_A$ to $T_A[1\dots {\tt MP1}]$ and resets $k_A, \widt{k}_A, \synca$ to 1 and 
$E_A, v_1, v_2$ to 0. Otherwise, if $v_2 \geq 0.2 k_A$, then Alice rewinds 
$T_A$ to $T_A[1\dots {\tt MP2}]$ and again resets $k_A, \widt{k}_A, \synca$ to 1 
and $E_A, v_1, v_2$ to 0.

If Alice has not made a meeting point transition, then Alice checks whether 
$E_A \geq 0.2 k_A$. If so, Alice undergoes an \emph{error transition}, in which 
she simply resets $k_A, \widt{k}_A, \synca$ to 1 and $E_A, v_1, v_2$ to 0 
(without modifying $T_A$).
 
Finally, if $k_A = \widt{k}_A \geq 2$ but Alice has not made any transition, 
then she simply resets $v_1, v_2$ to 0.

\begin{remark}
The idea behind meeting point transitions is that if the transcripts $T_A$ 
and $T_B$ have not diverged
too far, then there is a common meeting point up to which the
transcripts of Alice and Bob agree. Thus, during the control information of each
iteration, both Alice and Bob send hash values of their two meeting points in
the hope that there is a match. For a given scale $\widt{k}_A$, there are
$\widt{k}_A$ hash comparisons that are generated. If at least a constant
fraction of these comparisons result in a match, then Alice decides to
backtrack and rewind her transcript to the relevant meeting point. This ensures
that in order for an adversary to cause Alice to backtrack incorrectly, he must
corrupt the control information in a constant fraction of iterations.
\end{remark}


\subsection{Pseudocode}
We are now ready to provide the pseudocode for the protocol $\pienc$, which
follows the high-level description outlined in Section~\ref{subsec:highlevel}
and is shown in Figure~\ref{fig:oblivious}. The pseudocode for the helper
functions ${\tt AliceControlFlow}$, ${\tt AliceUpdateSyncStatus}$, 
${\tt AliceUpdateControl}$, \\
${\tt AliceDecodeControl}$, ${\tt AliceAdvanceBlock}$, ${\tt 
AliceUpdateEstimate}$, and ${\tt AliceRollback}$ for Alice 
is also displayed. Bob's functions ${\tt BobControlFlow}$, 
${\tt BobUpdateSyncStatus}$, ${\tt BobUpdateControl}$, ${\tt 
BobDecodeControl}$, ${\tt BobAdvanceBlock}$, ${\tt BobUpdateEstimate}$, and 
${\tt BobRollback}$ are almost identical, except that ``A'' subscripts are 
replaced with ``B'' and are thus omitted. Furthermore, the function 
${\tt InitializeSharedRandomness}$ is the same for Alice and Bob.
\begin{figure}[htp]
\begin{center}
\begin{tikzpicture}
\tikzstyle{every node}=[font=\small]
\def \ftw {15cm}
\def \tw {6.0cm}
\def \halfw {3.0in}
\def \stw {2cm}
\def \a {0cm}
\def \b {8.7cm}
\def \bleft {2.2in}
\def \bright {6.6cm}
\def \mid{5.25cm}
\def \aleft {\a - \tw/2}
\def \aright {\a + \tw/2}

\def \y {-1.1cm}
\node[draw, align=center] at (\mid,\y) {
{\bf Global parameters}\\[-0.1in]
\begin{minipage}{0.98\linewidth}
\begin{align*}
 b' = b + 2c\log(1/\epsilon')
\qquad &\piblk = \text{$B$-blocked simulating protocol for $\Pi$ (see 
Lemma~\ref{lem:blocking})}
\\[-0.04in]
\Niter = \frac{n'}{b}(1 
+ \Theta(\epsilon\log(1/\epsilon))) \qquad &l' = 
\Theta(n\epsilon'\,\mathrm{polylog}(1/\epsilon')) \\[-0.04in]
\epsilon' = \epsilon^2 \qquad &\chash: \{0,1\}^{\Theta(\log(1/\epsilon'))} \to 
\{0,1\}^{\Theta(\log(1/\epsilon'))} \text{ (see 
Section~\ref{subsubsec:controlencode})}\\[-0.04in]
b = s = \Theta(1/\epsilon') \qquad 
&\cexch: \{0,1\}^{l'} \to \{0,1\}^{10\epsilon\Niter b'} \text{ (see 
Section~\ref{subsec:randexch})}\\
B = sb \qquad 
&\crateless:\{0,1\}^B \to \{0,1\}^{2B} \text{ (see 
Lemma~\ref{lem:ecc})}
\end{align*}
\end{minipage}
};

\def \y {-3.7cm}
\node[draw] at (0.8cm,\y) {\bf Alice};
\node[draw] at (9.5,\y) {\bf Bob};

\def \y {-4.55cm}
\node[align=center] at (\mid, \y) {----------------- {\bf Random string
exchange} 
----------------- };

\def \y {-5.0cm}
\node[draw, below right, text width=2.6in] (arstr) at (\aleft,\y) {
Choose a random string $\shrand\in\{0,1\}^{l'}$\\
$w \gets \cexch(\shrand)$\\
};

\def \y {-6.6cm}
\node[draw, below right, text width=2.6in] (brstr) at (\bright,\y)
{
$w' \gets \text{nearest codeword of $\cexch$ to $\widetilde{w}$}$\\
$\shrand \gets (\cexch)^{-1}(w')$\\
};

\def \y {-8.1cm}
\node at (\mid, \y) {----------------- {\bf Initialization} 
----------------- };

\def \y {-8.6cm}
\node[draw, below right, text width=\halfw] (ainit) at (\aleft,\y) {
$T_A \gets \emptyset$\\
$x \gets nil$\\
$k_A, \widt{k}_A, c_A, \synca \gets 1$\\
$E_A, v_{1}, v_{2}, j_A, \spka, a, m, {\tt MP1}, {\tt MP2} \gets 0$ 
\vspace{2mm}\\
{\tt InitializeSharedRandomness()} \vspace{2mm}\\
{\bf if} Alice speaks in the first block of $\piblk$ {\bf then}\\
\hspace{2mm} $\spka\gets 1$\\
\hspace{2mm} $x \gets $ contents of first block of $\piblk$\\
\hspace{2mm} $y = y_0 \circ y_1 \circ \cdots \circ y_{2s-1} \gets 
\crateless(x)$\\
{\bf end if}
};
\node[draw, below right, text width=\halfw] (binit) at (\bleft,\y) {
$T_B \gets \emptyset$\\
$x \gets nil$\\
$k_B, \widt{k}_B, c_B, \syncb \gets 1$\\
$E_B, v_{1}, v_{2}, j_B, \spkb, a, m, {\tt MP1}, {\tt MP2} \gets 0$ 
\vspace{2mm}\\
{\tt InitializeSharedRandomness()} \vspace{2mm}\\
{\bf if} Bob speaks in the first block of $\piblk$ {\bf then}\\
\hspace{2mm} $\spkb\gets 1$\\
\hspace{2mm} $x \gets $ contents of first block of $\piblk$\\
\hspace{2mm} $y_0 \circ y_1 \circ \cdots \circ y_{2s-1} \gets \crateless(x)$\\
{\bf end if}
};

\def \y {-14.0cm}
\node at (\mid, \y) {----------------- {\bf Block transmission (repeat $\Niter$ 
times)} 
----------------- };

\def \y {-14.45cm}
\node[draw, below right, text width=\halfw] (ablock) at (\aleft,\y) {
{\tt AliceUpdateControl()}\vspace{2mm}\\
Send ${\bf r}[i]$ in slot $z_{m, i}^A$ for $i=1,\dots, (b'-b)/2$\\
Listen during slots $\widetilde{z}_{m,i}^B$ for $i=1,\dots, 
(b'-b)/2$ and write bits to $\widt{\bf r}$ \vspace{2mm}\\

{\bf if} $\synca = 1$ {\bf and } $\spka = 1$ {\bf then}\\
\hspace{2mm} Send the bits of $y_{m \bmod 2s}$ in the $b$\\
\hspace{3mm} remaining slots\\
{\bf else}\\
\hspace{2mm} Listen during the $b$ remaining slots and \\
\hspace{3mm} store as $g_A$\\
{\bf end if} \vspace{2mm}\\
{\tt AliceControlFlow()}
};
\node[draw, below right, text width=\halfw] (ablock) at (\bleft,\y) {
{\tt BobUpdateControl()}\vspace{2mm}\\
Send ${\bf r}[i]$ in slot $z_{m, i}^B$ for $i=1,\dots, (b'-b)/2$\\
Listen during slots $\widetilde{z}_{m,i}^A$ for $i=1,\dots, 
(b'-b)/2$ and write bits to $\widt{\bf r}$ \vspace{2mm}\\

{\bf if} $\syncb = 1$ {\bf and} $\spkb = 1$ {\bf then}\\
\hspace{2mm} Send the bits of $y_{m\bmod 2s}$ in the $b$\\
\hspace{3mm} remaining slots\\
{\bf else}\\
\hspace{2mm} Listen during the $b$ remaining slots and \\
\hspace{3mm} store as $g_B$\\
{\bf end if} \vspace{2mm}\\
{\tt BobControlFlow()}
};

\def \y {-20.9cm}
\node at (\mid, \y) {----------------- {\bf End of repeat} -----------------};


\draw[-latex] (arstr.south east) -- (brstr.north west) node[above, sloped, near 
start] 
{\hspace{2mm} $w$ \hspace*{5mm}} node[above, sloped, near end] 
{\hspace*{5mm}$\widetilde{w}$\hspace{2mm}};
\end{tikzpicture}
\vspace*{-6mm}
\caption{Encoded protocol $\pienc$ for tolerating oblivious adversarial errors.}
\label{fig:oblivious}
\end{center}
\end{figure}
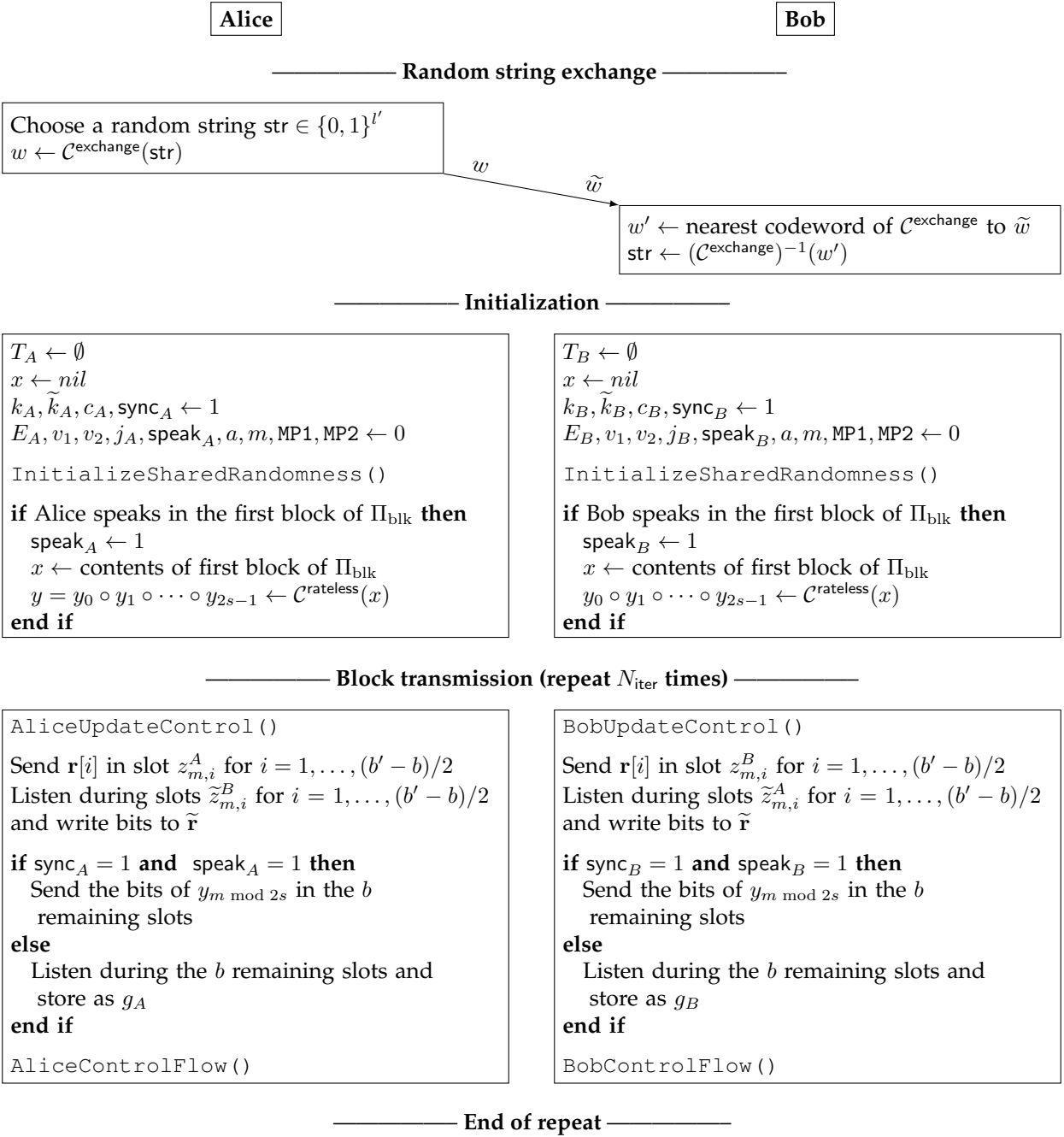

\begin{algorithm}[hp]
\caption{Procedure for Alice to process received data bits and control
info from a mini-block}
\begin{algorithmic}[1]
\Function{AliceControlFlow}{}

  \vspace{2mm}
  \noindent $\rhd$ {\bf \underline{Update phase}}:
  \vspace{2mm}
  \State $\widt{\ctrlinfo}_B^{(m)} \gets$ \Call{AliceDecodeControl}{}
  \If{$\widt{\ctrlinfo}_B^{(m)} \neq \perp$}
    \State $\left(\widt{h}_{B,c}^{(m)}, \widt{h}_{B,x}^{(m)}, 
\widt{h}_{B,k}^{(m)}, \widt{h}_{B,T}^{(m)}, \widt{h}_{B, {\tt MP1}}^{(m)}, 
\widt{h}_{B, {\tt MP2}}^{(m)}, \widt{j}_B, \rcvsyncb \right) \gets 
\widt{\ctrlinfo}_B^{(m)}$

    \vspace{3mm}

    \If{$\synca=0$}
      \If{$\widt{h}_{B,k}^{(m)} \neq {hash}_{B,k}^{(m)}(k_A)$ or $\rcvsyncb = 
1$}
	\State $E_A\gets E_A +1$
      \ElsIf{${hash}_{B,{\tt MP1}}^{(m)}(T_A[1\dots {\tt MP1}]) = 
  \widt{h}_{B,{\tt 
  MP1}}^{(m)}$ {\bf or} ${hash}_{B,{\tt MP2}}^{(m)}(T_A[1\dots {\tt MP1}]) = 
  \widt{h}_{B,{\tt MP2}}^{(m)}$}
	\State $v_1 \gets v_1 + 1$
      \ElsIf{${hash}_{B,{\tt MP1}}^{(m)}(T_A[1\dots {\tt MP2}]) = 
  \widt{h}_{B,{\tt MP1}}^{(m)}$ {\bf or} ${hash}_{B,{\tt MP2}}^{(m)}(T_A[1\dots 
  {\tt MP2}]) = 
  \widt{h}_{B,{\tt MP2}}^{(m)}$}
	\State $v_2 \gets v_2 + 1$
      \EndIf
      \EndIf
  \ElsIf{$\synca = 0$}
    \State $E_A\gets E_A+1$
  \EndIf
  
  \vspace{2mm}
  \If{$\synca = 0$}
    \State $k_A \gets k_A +1$
    \State $\tilde{k}_A \gets 2^{\left\lfloor 
\log_2 k_A \right\rfloor}$
  \EndIf
  \vspace{2mm}
  \State \Call{AliceUpdateSyncStatus}{}

  \vspace{3mm}
  
  \noindent $\rhd$ {\bf \underline{Transition phase}}:
  \vspace{2mm}

  \If{$k_A=\widt{k}_A \geq 2$ {\bf and} $v_1 \geq 0.2 k_A$}
    \State \Call{AliceRollback}{{\tt MP1}}
  \ElsIf{$k_A=\widt{k}_A \geq 2$ {\bf and} $v_2 \geq 0.2 k_A$}
    \State \Call{AliceRollback}{{\tt MP2}}
  \ElsIf{$k_A=\widt{k}_A \geq 2$ {\bf and} $E_A \geq 0.2 k_A$}
    \State $a \gets (m+1) \bmod 2s$
    \State $k_A,\widt{k}_A, \synca \gets 1$
    \State $E_A,v_1,v_2,j_A \gets 0$
  \ElsIf{$k_A=\widt{k}_A\geq 2$}
    \State $v_1,v_2\gets 0$
  \EndIf

  \vspace{3mm}
  
  \State ${\tt MP1} \gets \tilde{k}_A B \left\lfloor \frac{|T_A|}{\tilde{k}_A 
B} \right\rfloor$
  \State ${\tt MP2} \gets {\tt MP1} - \tilde{k}_A B$
  \State $m\gets m+1$

\EndFunction
\end{algorithmic}
\end{algorithm}

\begin{algorithm}[hp]
\caption{Procedure for Alice to update sync status}
\begin{algorithmic}[1]
\Function{AliceUpdateSyncStatus}{}
  \State $\synca\gets 0$
  \vspace{2mm}
  \If{$k_A = 1$}
    \If{$\widt{\ctrlinfo}_B^{(m)} \neq \,\,\perp$ {\bf and} 
$\widt{h}_{B,k}^{(m)} = {hash}_{B,k}^{(m)}(1)$}
      \If{$\rcvsyncb = 0$}
	\State $\synca\gets 1$; $j_A\gets 0$; $a\gets (m+1)\bmod{2s}$
      \ElsIf{${hash}_{B,c}^{(m)}(c_A) = \widt{h}_{B,c}^{(m)}$ {\bf and} 
${hash}_{B,T}^{(m)}(T_A) = \widt{h}_{B,T}^{(m)}$}
        \State $\synca\gets 1$
        \If{$\spka = 0$}
	  \If{$j_A\leq \widt{j}_B$}
	    \State \Call{AliceUpdateEstimate}{}
	  \Else
	    \State $j_A\gets 0$; $a\gets (m+1)\bmod{2s}$
	  \EndIf
	\Else
	  \State $j_A\gets \min\{j_A+1,2s\}$
        \EndIf
      \ElsIf{$\spka = 1$ {\bf and} ${hash}_{B,c}^{(m)}(c_A+1) = 
\widt{h}_{B,c}^{(m)}$ {\bf and } ${hash}_{B,T}^{(m)}(T_A \circ x) = \widt{h}_{B, 
T}^{(m)}$}
        \State $\synca\gets 1$
        \State \Call{AliceAdvanceBlock}{}
      \ElsIf{Bob speaks in block 
$(c_A-1)$ of $\piblk$ {\bf and} 
${hash}_{B,c}^{(m)}(c_A-1) = \widt{h}_{B,c}^{(m)}$ {\bf and}
${hash}_{B,T}^{(m)}(T_A [1\dots (c_A-2)B]) = \widt{h}_{B,T}^{(m)}$ {\bf and} 
${hash}_{B,x}^{(m)}(T_A [((c_A-2)B+1)\dots (c_A-1) B]) = 
\widt{h}_{B,x}^{(m)}$}
	\State $\synca \gets 1$
	\If{$\spka = 0$}
	  \State $j_A \gets 0$; $a\gets (m+1)\bmod {2s}$
	\Else
	  \State $j_A \gets \min\{j_A+1, 2s\}$
	\EndIf
      \EndIf
    \ElsIf{$\widt{\ctrlinfo}_B^{(m)} = \,\,\perp$}
      \State $\synca\gets 1$
      \If{$\spka = 0$}
        \If{$j_A \neq 0$}
	  \State \Call{AliceUpdateEstimate}{}
        \Else
	  \State $a\gets (m+1)\bmod{2s}$
	\EndIf
      \Else
	\State $j_A\gets \min\{j_A+1,2s\}$	
      \EndIf
    \EndIf
  \EndIf
\EndFunction
\end{algorithmic}
\end{algorithm}

\begin{algorithm}[htp]
\caption{Procedure for Alice to update control information}
\begin{algorithmic}[1]
\Function{AliceUpdateControl}{}
  \State \begin{varwidth}[t]{\linewidth - 0.5in} $\ctrlinfo_A^{(m)} \gets 
({hash}_{A,m}(c_A), 
{hash}_{A,x}^{(m)}(x), {hash}_{A,k}^{(m)}(k_A), {hash}_{A,T}^{(m)}(T_A), 
{hash}_{A,{\tt MP1}}^{(m)}(T_A[1 \dots {\tt MP1}]),$ ${hash}_{A,{\tt 
MP2}}^{(m)}(T_A[1 \dots {\tt MP2}]), j_A, 
\synca)$  \end{varwidth}
  \State ${\bf r} \gets \chash\left(\ctrlinfo_A^{(m)} \circ 
{hash}_{A,\mathrm{ctrl}}^{(m)}\left(\ctrlinfo_A^{(m)}\right)\right) \oplus 
V_A^{(m)}$
\EndFunction
\end{algorithmic}
\end{algorithm}

\begin{algorithm}[htp]
\caption{Procedure for Alice to decode control information sent by Bob}
\begin{algorithmic}[1]
\Function{AliceDecodeControl}{}
  \State ${\bf z} \gets $ decoding of ${\bf \widt{r}} \oplus V_B^{(m)}$ under 
$\chash$ (inverse of $\chash$ applied to nearest codeword)
  \State ${\bf z^c} \circ {\bf z^h} \gets {\bf z}$, 
where ${\bf z^c}$ has length $(b'-b)/2$
  \vspace{3mm}
  \If{${hash}_{B,\mathrm{ctrl}}^{(m)}({\bf z^c}) = {\bf z^h}$}
    \State \Return ${\bf z^c}$ 
  \Else
    \State \Return $\perp$
  \EndIf
\EndFunction
\end{algorithmic}
\end{algorithm}

\begin{algorithm}[htp]
\caption{Procedure for Alice to advance the block index and prepare for 
future transmissions}
\begin{algorithmic}[1]
\Function{AliceAdvanceBlock}{}
  \If{$\spka = 1$}
    \State $T_A \gets T_A \circ x$
  \Else
    \State $T_A \gets T_A \circ \widt{x}$
  \EndIf
  
  \vspace{3mm}
  \State $c_A \gets c_A + 1$
  \State $j_A \gets 0$
  \vspace{3mm}
  
  \If{Alice speaks in block $c_A$ of $\piblk$}
    \State $\spka \gets 1$
    \State $x\gets $ contents of block $c_A$ of $\piblk$
    \State $y = y_0 \circ y_1 \circ \cdots \circ y_{2s-1} \gets \crateless(x)$
  \Else
    \State $\spka \gets 0$
    \State $a \gets (m+1) \bmod 2s$
    \State $x\gets nil$
  \EndIf
\EndFunction
\end{algorithmic}
\end{algorithm}

\begin{algorithm}
\begin{algorithmic}[1]
\caption{Procedure for Alice to update her estimate of the contents of the 
current block based on past data blocks}
\Function{AliceUpdateEstimate}{}
  \State $\widt{g}_{j_A} \gets g_A$
  \State $j_A \gets j_A + 1$
  
  \vspace{3mm}
  
  \If{$j_A > s$}
    \State $\widt{x}\gets $ result after decoding $(\widt{g}_0, 
\widt{g}_1, \dots, \widt{g}_{j_A - 1})$ via the nearest codeword in 
$\crateless_{a,
j_A}$
    \If{${hash}_{B,x}^{(m)}(\widt{x}) = \widt{h}_{B,x}^{(m)}$}
      \State \Call{AliceAdvanceBlock}{}
    \ElsIf{$j_A = 2s$}
      \State $j_A \gets 0$
      \State $a \gets (m+1)\bmod 2s$
    \EndIf
  \EndIf
\EndFunction
\end{algorithmic}
\end{algorithm}

\begin{algorithm}
\begin{algorithmic}[1]
\caption{Procedure for Alice to backtrack to a previous meeting point}
\Function{AliceRollback}{{\tt MP}}
  \State $T_A \gets T_A[1\dots {\tt MP}]$
  \State $c_A \gets \frac{{\tt MP}}{B} + 1$
  \State $k_A, \widt{k}_A, \synca \gets 1$
  \State $E_A,v_1,v_2, j_A \gets 0$
  
  \vspace{3mm}
  
  \If{Alice speaks in block $c_A$ of $\piblk$}
    \State $\spka \gets 1$
    \State $x\gets $ contents of block $c_A$ of $\piblk$
    \State $y = y_0 \circ y_1 \circ \cdots \circ y_{2s-1} \gets \crateless(x)$
  \Else
    \State $\spka \gets 0$
    \State $a \gets (m+1) \bmod 2s$
    \State $x\gets nil$
  \EndIf
\EndFunction
\end{algorithmic}
\end{algorithm}

\begin{algorithm}[hp]
\caption{Procedure for Alice and Bob to use exchanged random string to 
initialize hash functions, information hiding mechanism, and encoding functions 
for control information}
\begin{algorithmic}[1]
\Function{InitalizeSharedRandomness}{}
  \State $p\gets \Theta(\log(1/\epsilon'))$
  \State $\delta \gets 2^{-\Theta(\Niter\cdot p)}$
  \State $L \gets \Theta(\Niter n\log(1/\epsilon'))$
  \State \begin{varwidth}[t]{\linewidth - 0.5in} Let $\shrand = \locrand \circ 
\shremain$, where $\locrand$ is of
length $\Theta(\Niter\cdot \log^2(1/\epsilon'))$ and $\shremain$ is of length 
$\Theta(\log(L/\delta))$ \end{varwidth}
  \State \begin{varwidth}[t]{\linewidth - 0.5in} $S \gets $ $\delta$-biased
length $L$ pseudorandom string derived 
from $\shremain$ (via the biased sample space of \cite{NaorNaor}) \end{varwidth}

  \vspace{2mm}
  \noindent $\rhd$ {\bf \underline{Generate locations for information hiding
in each iteration}}:
  \vspace{2mm}

  \For{$i = 0$ to $\Niter-1$}
    \State \begin{varwidth}[t]{\linewidth - 0.5in} Choose $z_{i,1}^A,
z_{i,2}^A, \dots, z_{i, (b'-b)/2}^A, z_{i,1}^B,
z_{i,2}^B, \dots, z_{i, (b'-b)/2}^B$ to be distinct numbers in
$\{1,2,\dots,b'\}$ using $O(\log^2 (1/\epsilon'))$ fresh random bits from
$\locrand$ \end{varwidth}
  \EndFor
  
  \vspace{2mm}
  \noindent $\rhd$ {\bf \underline{Set up parameters for encoding control 
information during each iteration}}
  \vspace{2mm}
  
  \For{$i = 0$ to $\Niter-1$}
    \State $V_A^{(i)} \gets $ $(b'-b)/2$ fresh random bits 
from $\locrand$
    \State $V_B^{(i)} \gets $ $(b'-b)/2$ fresh random bits 
from $\locrand$
    \State \begin{varwidth}[t]{\linewidth - 0.5in} Initialize 
${hash}_{A,\mathrm{ctrl}}^{(i)}, {hash}_{B,\mathrm{ctrl}}^{(i)}$ to an inner 
product hash function with output length $\Theta(\log(1/\epsilon'))$ and 
seed fixed as $\Theta(\log(1/\epsilon'))$ fresh random bits from $\locrand$ 
\end{varwidth}
  \EndFor

  \vspace{2mm}
  \noindent $\rhd$ {\bf \underline{Initialize hash functions for control 
information in each iteration}}:
  \vspace{2mm}

  \For{$i = 0$ to $\Niter-1$}
    \State \begin{varwidth}[t]{\linewidth - 0.5in} Initialize 
${hash}_{A,c}^{(i)}$, ${hash}_{A,x}^{(i)}$, ${hash}_{A,k}^{(i)}$, 
${hash}_{A,T}^{(i)}$, ${hash}_{A,{\tt MP1}}^{(i)}$, ${hash}_{A,{\tt 
MP2}}^{(i)}$, ${hash}_{B,c}^{(i)}$, ${hash}_{B,x}^{(i)}$, ${hash}_{B,k}^{(i)}$, 
${hash}_{B,T}^{(i)}$, ${hash}_{B,{\tt MP1}}^{(i)}$, ${hash}_{B,{\tt 
MP2}}^{(i)}$ to be inner product hash functions with output length 
$\Theta(\log(1/\epsilon'))$ and seed fixed using fresh random bits from 
$S$ \end{varwidth}
  \EndFor
\EndFunction
\end{algorithmic}
\end{algorithm}

\subsection{Analysis of Coding Scheme for Oblivious Adversarial Channels}
Now, we show that the coding scheme presented in Figure~\ref{fig:oblivious} 
allows one to tolerate an error fraction of $\epsilon$ under an oblivious 
adversary with high probability.

\subsubsection{Protocol States and Potential Function}
Let us define states for the encoded protocol $\pienc$. First, we define
\[
 \ell^+ = \left\lfloor \max\{\ell' \in [1, \min\{|T_A|, |T_B|\}]:
T_A[1\dots \ell'] = T_B[1\dots \ell']\right\rfloor, \quad\quad \ell^- =
|T_A|+|T_B| - 2\ell^+.
\]
In other words, $\ell^+$ is the length of the longest
common prefix of the transcripts $T_A$ and $T_B$, while $\ell^-$ is the total
length of the parts of $T_A$ and $T_B$ that are not in the common
prefix. Also recall that $\delta_{s+1}, \delta_{s+2}, \dots, 
\delta_{2s}$ are defined as in Lemma~\ref{lem:ecc}. Furthermore, we define 
$\delta_0, \delta_1, \dots, \delta_s = 0$ for convenience.

Now we are ready to define states for the protocol $\pienc$ as its execution 
proceeds.

\begin{definition}\label{def:state}
At the beginning of an iteration (the start of the code block in 
Figure~\ref{fig:oblivious} that is repeated $\Niter$ times), the protocol is 
said to be in one of three possible states:
\begin{itemize}
 \item \textbf{Perfectly synced} state: This occurs if $\synca=\syncb=1$, $k_A 
= k_B = 1$, $\ell^- 
= 0$, $c_A=c_B$, and $j_A \geq j_B$ if Alice is the sender in block $c_A=c_B$ 
of $\piblk$ (resp. $j_B \geq j_A$ if Bob is the sender in $B$-block $c_A=c_B$ 
of $\piblk$). In this case, we also define $j = \min\{j_A, j_B\}$.
 \item \textbf{Almost synced} state: This occurs if $\synca=\syncb=1$, $k_A = 
k_B = 1$, and one of the following holds:
\begin{enumerate}
 \item $\ell^-=B$, $c_B = c_A + 1$, and $T_B = T_A\circ w$, where $w$ 
represents the contents of the $c_A$-th $B$-block of $\piblk$. In this case, we 
define $j = j_B$.
 \item $\ell^-=B$, $c_A = c_B + 1$, and $T_A = T_B\circ w$, where $w$ 
represents the contents of the $c_B$-th $B$-block of $\piblk$. In this case, we 
define $j = j_A$.
 \item $\ell^-=0$, $c_A = c_B$, $j_B > j_A$, and Alice speaks in $B$-block 
$c_A=c_B$ of $\piblk$. In this case, we define $j = j_B$.
 \item $\ell^-=0$, $c_A = c_B$, $j_A > j_B$, and Bob speaks in $B$-block 
$c_A=c_B$ of $\piblk$. In this case, we define $j = j_A$.
\end{enumerate}
 \item \textbf{Unsynced} state: This is any state that does not fit into the 
above two categories.
\end{itemize}
\end{definition}
\noindent We also characterize the control information sent by each party 
during an iteration based on whether/how it is corrupted.
\begin{definition}
For any given iteration, the encoded control information sent by a party is 
categorized as one of the following:
\begin{itemize}
 \item \textbf{Sound control information}: If a party's unencoded control 
information for an iteration is decoded correctly by the other party (i.e., the 
output of $\contdec$ correctly retrieves the intended transmission), and no 
hash collisions (involving the hashes contained in the control information 
$\widt{\ctrlinfo}_A^{(m)}$ or $\widt{\ctrlinfo}_B^{(m)}$) occur, then the 
(encoded) control information is considered \emph{sound}.

 \item \textbf{Invalid control information}: If the attempt to decode a party's 
unencoded control information by the other party results in a failure (i.e., 
$\contdec$ outputs $\perp$), then the (encoded) control information is 
considered \emph{invalid}.

 \item \textbf{Maliciously corrupted control information}: If a party's 
control information is decoded incorrectly (i.e., $\contdec$ does not output 
$\perp$, but the output does not retrieve the intended transmission) or a hash 
collision (involving the hashes contained in the control information 
$\widt{\ctrlinfo}_A^{(m)}$ or $\widt{\ctrlinfo}_B^{(m)}$) occurs, then the 
(encoded) control information is considered \emph{maliciously corrupted}.
\end{itemize}
\end{definition}
\noindent Next, we wish to define a potential function $\Phi$ that depends on 
the current state in the encoded protocol. Before we can do so, we define a few 
quantities:
\begin{definition}
 Suppose the protocol is in a perfectly synced state. Then, we define the 
quantities $\curerr$ and $\inv$ as follows:
\begin{itemize}
 \item $\curerr$ is the total number of data (non-control information) bits 
that have been corrupted 
during the last $j$ iterations.
 \item $\inv$ is the number of iterations among the last $j$ iterations for 
which the control information of at least one party was invalid or maliciously 
corrupted.
\end{itemize}
\end{definition}
\begin{definition}
 Suppose the protocol is in an unsynced state. Then, we define $\malA$ as 
follows: At the start of $\pienc$, we initialize $\malA$ to 0. Whenever an 
iteration occurs from a state in which $\synca = 0$, such that either Alice's 
or Bob's control information during that iteration is maliciously corrupted, 
$\malA$ increases by 1 at the end of line 21 of ${\tt AliceControlFlow}$ during 
that iteration. Moreover, whenever Alice undergoes a \emph{transition} (i.e., 
one of the ``if'' conditions in lines 22-29 of ${\tt AliceControlFlow}$ is 
true), $\malA$ resets to 0.

The variable $\malB$ is defined in the obvious analagous manner.
\end{definition}
\begin{definition}
 For the sake of brevity, a variable $\mathrm{var}_{AB}$ will denote
$\mathrm{var}_A + \mathrm{var}_B$ (e.g., $k_{AB} = k_A+k_B$ and $E_{AB} = E_A 
+ E_B$).
\end{definition}

\noindent Now, we are ready to define the potential function $\Phi$.
\begin{definition}
Let $C_0, C_1, C_2, C_3, C_4, C_5, C_6, C_7, \cinv, \cmal, C, D > 0$ be 
suitably chosen constants (to be determined by Lemmas~\ref{lem:synctrans}, 
\ref{lem:asynctrans}, \ref{lem:unsynctrans} and Theorem~\ref{thm:nonrate}). 
Then, we define the potential function $\Phi$ associated with the execution of 
$\pienc$ according to the state of the protocol (see 
Definition~\ref{def:state}): 
\begin{align*}
 \Phi = \begin{cases}
         \ell^+(1+ C_0 H(\epsilon)) + (jb -
C\cdot\curerr\cdot\log(1/\epsilon)) - Db\cdot\inv \quad &\text{perfectly
synced}\\
\max\{\ell_A, \ell_B\}\cdot(1 + C_0 H(\epsilon)) - (j+1)b\quad &\text{almost
synced}\\
\ell^+ (1 + C_0 H(\epsilon)) - C_1 \ell^- + b(C_2 k_{AB} - C_3 E_{AB}) \quad 
&\text{unsynced, $(k_A, \synca) = (k_B, \syncb)$} \\
\qquad - 2C_7 B\,\malAB - Z_1
\quad
&\ \\
\ell^+ (1 + C_0 H(\epsilon)) - C_1 \ell^- + bC_5(-0.8 k_{AB} + 0.9 E_{AB})\quad 
&\text{unsynced, $(k_A, \synca) \neq (k_B, \syncb)$} \\
\qquad - C_7  B\,\malAB - Z_2 \quad &\ \\
        \end{cases},
\end{align*}
where $Z_1$ and $Z_2$ are defined by:
\[
  Z_1 = \begin{cases} bC_4 \quad &\text{if $k_A = k_B = 1$ 
and $\synca=\syncb=1$}\\ 
\frac{1}{2}bC_4 \quad &\text{if $k_A = k_B = 1$ and $\synca=\syncb=0$}\\
0 \quad &\text{otherwise}
\end{cases},
 \]
and
\[
 Z_2 = \begin{cases}
bC_6 \quad &\text{if $k_A=k_B=1$}\\
0\quad &\text{otherwise}
  \end{cases}.
\]
\end{definition}

\subsubsection{Bounding Iterations with Invalid or Maliciously Corrupted 
Control Information}
We now prove some lemmas that bound the number of iterations that can have 
invalid or maliciously corrupted control information.

\begin{lemma}\label{lem:whpcontrol}
 If the fraction of errors in a mini-block is $O(1)$, say, $< \frac{1}{20}$, 
then 
with probability at least $1 - {\epsilon'}^{2}$, 
both parties can correctly decode and verify the control symbols sent in the
block.
\end{lemma}
\begin{proof}
Let $\nu < 1/20$ be the fraction of errors in a mini-block. Recall that Alice's
control information in the mini-block consists of $c\log(1/\epsilon')$ randomly
located bits. Let $X$ be the number of these control bits that are corrupted.
Note that $\E[X] = \nu c\log(1/\epsilon')$. Now, since the control information
is protected with an error correcting code of distance $c\log(1/\epsilon')/4$,
we see that Bob can verify and correctly decode Alice's control symbols as long
as $X < c\log(1/\epsilon')/8$. Note that by the Chernoff bound,
\begin{align*}
 \Pr\left(X > c\log(1/\epsilon')/8\right) &\leq
e^{-\frac{\frac{c\log(1/\epsilon')}{8} - \frac{c\log(1/\epsilon')}{20}}{3}}\\
&\leq {\epsilon'}^{c/40},
\end{align*}
which is $< \epsilon'^2 / 2$ for a suitable constant $c$. Similarly, the
probability that Alice fails to verify and correctly decode Bob's control
symbols is $< \epsilon'^2/2$. Thus, the desired statement follows by a
union bound.
\end{proof}

\begin{lemma}\label{lem:invbound}
 With probability at least $1 - 2^{-\Omega(\epsilon'\Niter)}$, the 
number of iterations in which some party's control information is invalid but 
neither party's control information is maliciously corrupted is $O(\epsilon 
\Niter)$.
\end{lemma}
\begin{proof}
First of all, consider the number of iterations of $\pienc$ for which the 
fraction of errors within the iteration is at least $1/20$. Since the total 
error fraction throughout the protocol is $\epsilon$, we know that at at most 
$20\epsilon\Niter$ iterations have such an error fraction.

Next, consider any ``low-error'' iteration in which the error fraction is less 
than $1/20$. By Lemma~\ref{lem:whpcontrol}, the probability that control 
information of some party is invalid (but neither party's control information 
is maliciously corrupted) is at most $\epsilon'^2$. Then, by the Chernoff 
bound, the number of ``low-error'' iterations with invalid control information 
is at most $(\epsilon'^2 + \epsilon')\Niter = O(\epsilon' \Niter)$ with 
probability at least $1-2^{-\Omega(\epsilon' \Niter)}$.

It follows that with probability at least $1-2^{-\Omega(\epsilon' \Niter)}$, 
the total number of iterations with invalid control information (but not 
maliciously corrupted control information) is $O(\epsilon\Niter)$.
\end{proof}

\begin{lemma}\label{lem:malbound}
With probability at least $1 - 2^{-\Omega(\epsilon'^2 \Niter)}$, the 
number of iterations in which some party's control information is 
maliciously corrupted is at most $O(\epsilon'^2 \Niter)$.
\end{lemma}
\begin{proof}
Suppose a particular party's control information is maliciously corrupted 
during a certain iteration (say, the $m^\text{th}$ iteration). Without loss of 
generality, assume Alice's control 
information is maliciously corrupted. Then, we must have one of 
the following:
\begin{enumerate}
 \item The number of corrupted bits in the encoded control information of Alice 
is $> \frac{1}{8} \left(\frac{b'-b}{2}\right)$, i.e., the fraction of 
control information bits that is corrupted is greater than $\frac{1}{8}$.

\item The number of corrupted bits in the encoded control information of Alice 
is $< \frac{1}{8}\left(\frac{b'-b}{2}\right)$, but a hash collision occurs for 
one of $h_{A,c}^{(m)}$, $h_{A,x}^{(m)}$, $h_{A,k}^{(m)}$, $h_{A,T}^{(m)}$, 
$h_{A,{\tt MP1}}^{(m)}$, $h_{A, {\tt MP2}}^{(m)}$.
\end{enumerate}
Note that by Property~\ref{prop:encdec}, case (1.) happens with 
probability at most
\[
2^{-\Theta(\log(1/\epsilon'))} \leq \epsilon'^2,
\]
for suitable constants.

Next, we consider the probability that case (2.) occurs. By 
Property~\ref{prop:collisionbiased}, we have that the probability of a hash 
collision any specific quantity among $c_A$, $x$, $k_A$, $T_A$, $T_A[1, {\tt 
MP1}]$, $T_A[1, {\tt MP2}]$ is at most $2^{-\Theta(\log(1/\epsilon'))} + 
2^{-\Theta(\Niter \log(1/\epsilon'))} \leq \epsilon'^2$ for appropriate 
constants. Thus, by a simple union bound, the probability that any one 
of the aforementioned quantities has a hash collision is at most $6\epsilon'^2$.

A simple union bound between the two events shows that the probability that 
Alice's control information in a given iteration is maliciously corrupted is at 
most $7\epsilon'^2$. Similarly, the probability that Bob's control information 
in a given iteration is maliciously corrupted is also at most $7\epsilon'^2$. 
Hence, the desired claim follows by the Chernoff bound (recall that there is 
limited independence, due to the fact that we use pseudorandom bits to seed 
hash functions, but this is not a problem due to our choice of parameters (see 
Section~\ref{subsec:controlinfo})).
\end{proof}

\subsubsection{Evolution of Potential Function During Iterations}
We now wish to analyze the evolution of the potential function $\Phi$ as the 
execution of the protocol proceeds. First, we define some notation that will 
make the analysis easier:
\begin{definition}
 Suppose we wish to analyze a variable $\mathrm{var}$ over
the course of an iteration. For the purpose of Lemmas~\ref{lem:synctrans}, 
\ref{lem:asynctrans}, and \ref{lem:unsynctrans}, we let $\mathrm{var}$
denote the value of the variable at the start of the iteration (the start of 
the code block in Figure~\ref{fig:oblivious} that is repeated $\Niter$ times). 
Moreover, we let $\mathrm{var}'$ denote the value of the variable
just after the ``update phase'' of the iteration (lines 2-21 of ${\tt
AliceControlFlow}$ and ${\tt BobControlFlow}$), while we will let
$\mathrm{var}''$ denote the value of the variable at the end of the iteration 
(at the end of the execution of ${\tt AliceControlFlow}$ and ${\tt 
BobControlFlow}$).

Moreover, we will use the notation $\Delta\mathrm{var}$ to denote 
$\mathrm{var}'' - \mathrm{var}$, i.e., the change in the variable over the 
course of an iteration. For instance, $\Delta\Phi = \Phi'' - \Phi$.
\end{definition}

\begin{definition}
During an iteration of $\pienc$, Alice is said to undergo a \emph{transition} 
if one of the ``if'' conditions in lines 22-29 of ${\tt AliceControlFlow}$ is 
true. The transition is called a \emph{meeting point transition} (or \emph{MP 
transition}) if either line 23 or line 25 is executed, while the transition is 
called an \emph{error transition} if lines 27-29 are executed. Transitions for 
Bob are defined similarly, except that one refers to lines in the corresponding 
${\tt BobControlFlow}$ function.
\end{definition}

\noindent Now, we are ready for the main analysis. Lemmas~\ref{lem:synctrans}, 
\ref{lem:asynctrans}, and \ref{lem:unsynctrans} prove lower bounds on the 
change in potential, $\Delta\Phi$, over the course of an iteration, depending 
on (1.) the state of the protocol prior to the iteration and (2.) whether/how 
control information is corrupted during the iteration.

\begin{lemma}\label{lem:synctrans}
Suppose the protocol is in a perfectly synced state at the beginning of an 
iteration. Then, the change in potential $\Phi$ over the course of the  
iteration behaves as follows, according to the subsequent state (at the end of 
the iteration):
\begin{enumerate}
\item If the subsequent state is perfectly synced or almost synced, then:
\begin{itemize}
\item If the control information received by both parties is sound, then 
$\Delta\Phi \geq b - Ct\cdot\log(1/\epsilon)$, where $t$ is the number of 
data (non-control) bits that are corrupted in the next iteration.

\item If the control information received by at least one party is invalid or 
maliciously corrupted, then $\Delta\Phi \geq -Ct\cdot\log(1/\epsilon) - (D-1)b 
\geq -Ct\cdot\log(1/\epsilon) - \min\{\cinv b,\cmal B\}$.
\end{itemize}

\item If the subsequent state is unsynced, then $\Delta\Phi \geq -\cmal B$.
\end{enumerate}
\end{lemma}

\begin{proof}
Assume that the protocol is currently in a perfectly synced state, and, without 
loss of generality, suppose that Alice is trying to send data bits 
corresponding to $c_A$-\text{th} $B$-block of $\piblk$ to Bob.

For the first part of the lemma statement, assume that the state after the next 
iteration is perfectly synced or almost synced. At the end of the iteration, Bob 
updates his estimate of what Alice is sending, and there are three cases:
\begin{itemize}
 \item \underline{Case 1}: Bob is still not able to decode the $c_A$-th 
$B$-block that Alice 
is sending, and $j_B$ does not reset to zero. In this case, it is clear that 
$j$ 
increases by 1, while $\curerr$ increases by $t$. Thus, $\Delta\Phi 
\geq b-Ct\cdot\log(1/\epsilon)$ if the control information received by both 
parties is sound, while $\Delta\Phi \geq -Ct\cdot\log(1/\epsilon) - (D-1)b$ 
otherwise (as $\inv$ increases by 1).

\item \underline{Case 2}: Bob is still not able to decode the $c_A$-th 
$B$-block that Alice is sending, but $j_B$ resets to 0 (after increasing to 
$2s$). Then, note that if both parties receive sound control information in the 
next iteration, we have
\[
 \Delta\Phi \geq (b-Ct\cdot\log(1/\epsilon)) + (Db\cdot\inv 
+ C(\curerr+t)\log(1/\epsilon) - 2B).
\]
Moreover, we must have $\curerr + t \geq \frac{1}{2} \delta_{2s} (2B) = 
\frac{1}{15}B$, which implies that
\[
 Db\cdot\inv + C(\curerr + t)\log(1/\epsilon) - 2B \geq 0,
\]
as desired (for suitably large $C$).

On the other hand, suppose some party receives invalid or maliciously corrupted 
control information in the next iteration. Then,
\[
 \Delta\Phi \geq (-Ct\cdot\log(1/\epsilon) - (D-1)b) + (Db\cdot(\inv+1) 
+ C(\curerr+t) \log(1/\epsilon) - 2B).
\]
Thus, to prove the lemma, it suffices to show 
\begin{equation}
 Db\cdot(\inv+1) + C(\curerr+t) \log(1/\epsilon) - 2B \geq 0. \label{eq:reset}
\end{equation}
Let $j_0$ be the last/most recent value of $j_B$ occurring after an 
iteration in which Bob receives sound control information (or 
$j_0 = 0$ if such an iteration did not occur). Thus, in the 
last $2s-j_0-1$ iterations, Bob has not received sound control information. 
This implies that $\inv \geq 2s-j_0-1$ and $\curerr \geq 
\frac{1}{2}\delta_{j_0}j_0 b$. Thus, we reduce (\ref{eq:reset}) to showing the 
following:
\begin{equation}
 D(2s-j_0)b + \frac{C}{2}\delta_{j_0}j_0 b \cdot 
\log(1/\epsilon) - 2B \geq 0. \label{eq:resetupdate}
\end{equation}
Note that if $j_0 \leq s$, then $\delta_{j_0} = 0$, and so the lefthand side of 
(\ref{eq:resetupdate}) is at least
\[
 Dsb - 2B = (D-2)B \geq 0,
\]
as desired. Hence, we now assume that $j_0 > s$. Then, by Lemma~\ref{lem:ecc}, 
$\delta_{j_0} \geq H^{-1}\left(\frac{j_0 - s}{j_0} - \frac{1}{4s}\right)$ 
(recall that $H^{-1}$ is the unique inverse of $H$ that takes values in
$[0,1/2]$). Thus, (\ref{eq:resetupdate}) reduces to showing
\begin{equation}
 \frac{C}{2}H^{-1}\left(\frac{j_0 - s}{j_0} - 
\frac{1}{4s}\right)\log(1/\epsilon) \geq D - \frac{2s(D-1)}{j_0}. 
\label{eq:redeq}
\end{equation}
Note that if $j_0 \leq \frac{D-1}{D}\cdot 2s$, then (\ref{eq:redeq}) is 
clearly true, as the righthand side of (\ref{eq:redeq}) is nonpositive.

If $j_0 > \frac{D-1}{D}\cdot 2s$, then note that the righthand side of 
(\ref{eq:redeq}) is at most 1 (since $j_0 \leq 2s$), while the lefthand side 
is at least
\begin{align*}
 \frac{C}{2} H^{-1}\left(1 - \dfrac{s}{\frac{D-1}{D}\cdot 2s} - 
\frac{\epsilon'}{4}\right)\log(1/\epsilon) &\geq \frac{C}{2} 
H^{-1}\left(\frac{D-2}{2(D-1)} - \frac{\epsilon'}{4}\right)\log(1/\epsilon)\\
&\geq 1.
\end{align*}

\item\underline{Case 3}: Bob manages to decode the $c_A$-th $B$-block and 
updates his transcript. Then, the protocol either transitions to an almost 
synced state or remains in a perfectly synced state (if Alice receives 
maliciously corrupted control information indicating that Bob has already 
advanced his transcript). Thus,
\[
 \Delta\Phi \geq (b-Ct\cdot\log(1/\epsilon)) + B(1+C_0 H(\epsilon)) + 
C (\curerr+t) \log(1/\epsilon) - (j+2)b + Db\cdot\inv,
\]
Hence, it suffices to show that
\begin{equation}
 B(1 + C_0 H(\epsilon)) + C(\curerr+t) \log(1/\epsilon) - (j+2)b + 
Db\cdot\inv \geq 0. 
\label{eq:nonnegpot}
\end{equation}

Note that $j \geq s$. Suppose $j_0$ is the last/most recent value of $j_B$ 
occurring after an iteration in which Bob receives sound control information
(or 
$j_0 = 0$ if such an iteration did not 
occur). Then, $\inv \geq j-j_0$. Hence, (\ref{eq:nonnegpot}) reduces to 
showing
\begin{equation}
 B(1+C_0 H(\epsilon)) + C\cdot\curerr' \cdot 
\log(1/\epsilon) - (j+2)b + Db(j-j_0) \geq 0. \label{eq:rednonneg}
\end{equation}
Note that if $j_0 \leq s$, then the lefthand side of (\ref{eq:rednonneg}) is at 
least
\begin{align*}
 B(1+C_0 H(\epsilon)) - (j+2)b + Db(j-s) &\geq B(1+C_0 H(\epsilon)) + (D-1)jb - 
DB - 2b\\
&\geq B(1+C_0 H(\epsilon)) + (D-1)B - DB - 2b\\
&\geq B(C_0 H(\epsilon) - 2\epsilon')\\
&\geq 0,
\end{align*}
as desired.

Now, assume $j_0 > s$. Let $\epsilon_0$ be the fraction of errors in the first 
$j_0 b$ data bits sent since Alice and Bob became perfectly synced (or since the 
last reset). Then,
\[
 \curerr' \geq \epsilon_0 j_0 b.
\]
Hence, the lefthand side of (\ref{eq:rednonneg}) is at least
\begin{equation}
 B(1+C_0 H(\epsilon)) + j_0 b(C\epsilon_0 \log(1/\epsilon) - 1) - 2b + 
(D-1)b(j-j_0). \label{eq:lhspot}
\end{equation}
Note that if $C\epsilon_0 \log(1/\epsilon) \geq 1$, then the above quantity is 
clearly nonnegative, as $B \geq b/\epsilon' \geq 2b$. Thus, let us assume that 
$C\epsilon_0 \log(1/\epsilon) < 1$. Now, recall from our choice of $\crateless$ 
and the fact that Bob had not successfully decoded the blocks sent by Alice 
before the current iteration, we have $\epsilon_0 \geq 
\frac{1}{2}\delta_{j_0}$, which implies that
\begin{align*}
 \frac{j_0 - s}{j_0} - \frac{1}{4s} = H(\delta_{j_0}) \leq H(2\epsilon_0).
\end{align*}
Hence,
\begin{align*}
 j_0 \leq \frac{s}{1-H(2\epsilon_0)-\frac{1}{4s}}.
\end{align*}
Now, (\ref{eq:lhspot}) is at least
\begin{align}
 &\ B(1+C_0 H(\epsilon)) + 
\frac{B(C\epsilon_0 \log(1/\epsilon)-1)}{1-H(2\epsilon_0)-\frac{1}{4s}} - 2b 
\nonumber\\ 
&\geq B(1+C_0 H(\epsilon)) + 
\frac{B(C\epsilon_0 \log(1/\epsilon)-1)}{1-H(2\epsilon_0)-\frac{\epsilon'}{4}} 
- 2b \nonumber \\
&\geq B\left(1 + C_0 H(\epsilon) - (1-C\epsilon_0 \log(1/\epsilon))\left(1 + 
H(2\epsilon_0) + \frac{\epsilon'}{4} + 2\left(H(2\epsilon_0) + 
\frac{\epsilon'}{4}\right)^2\right) - 2\epsilon'\right) \nonumber \\
&\geq B \left(1 + C_0 H(\epsilon) - 1 - H(2\epsilon_0) - \frac{\epsilon'}{4} - 
2 H(2\epsilon_0)^2 - \epsilon' H(2\epsilon_0) - \frac{\epsilon'^2}{8} + 
C\epsilon_0\log(1/\epsilon) - 2\epsilon'\right) \nonumber \\
&\geq B\left( C_0 H(\epsilon) - 4\epsilon' - 3H(2\epsilon_0) + 
C\epsilon_0 \log(1/\epsilon) \right). \label{eq:simppot}
\end{align}
Note that if $\epsilon_0 < \epsilon$, then (\ref{eq:simppot}) is bounded from 
below by
\begin{align*}
 B(C_0 H(\epsilon) - 4\epsilon' - 3 H(2\epsilon)) &\geq B\left((4 H(\epsilon) 
- 4\epsilon') + ((C_0-4)H(\epsilon) - 3H(2\epsilon))\right)\\
&\geq 0,
\end{align*}
since $H(\epsilon) \geq \epsilon \geq \epsilon'$, $C_0 \geq 10$, and 
$2H(\epsilon)\geq H(2\epsilon)$.

On the other hand, if $\epsilon_0 \geq \epsilon$, then (\ref{eq:simppot}) is 
bounded 
from below by
\[
B\left((4 H(\epsilon) - 4\epsilon') + (C\epsilon_0 \log(1/\epsilon_0) - 3 
H(2\epsilon_0))\right) \geq 0,
\]
as long as $C\geq 10$.
\end{itemize}
This completes the proof of the first part of the lemma.

Next, we prove the second part of the lemma. Assume that the protocol is 
currently in a perfectly synced state and that the subsequent state is 
unsynced. Then, note that the control information of at least one party must be 
maliciously corrupted. Observe that $k_A'' = k_B'' = 1$, and ${\ell^-}'' 
\leq 2B$, while $E_A'' = E_B'' = 0$. Thus, if $\synca'' = \syncb''$, then
\[
 \Delta\Phi \geq -jb- 2 C_1 B + 2b C_2 - bC_4 \geq -\cmal B,
\]
while if $\synca'' \neq \syncb''$, then
\[
 \Delta\Phi \geq -jb - 2 C_1 B - 1.6bC_5 - bC_6 \geq -\cmal B,
\]
since $jb \leq 2B$.
\end{proof}

\begin{lemma}\label{lem:asynctrans}
 Suppose the protocol is in an almost synced state at the beginning of an 
iteration. Then, the change in 
potential $\Phi$ over the course of the iteration behaves as follows, according 
to the control information received during the iteration:
\begin{itemize}
\item If the control information received by both parties is sound, then 
$\Delta\Phi \geq b$.

\item If the control information received by at least one party is invalid, but 
neither party's control information is maliciously corrupted, then the 
potential does not change, i.e., $\Delta\Phi \geq -b \geq -\cinv b$.

\item If the control information received by at least one party is maliciously 
corrupted, then $\Delta\Phi \geq -\cmal B$.
\end{itemize}
\end{lemma}
\begin{proof}
Assume the protocol lies in an almost synced state. We consider the following 
cases, according to the subsequent state in the protocol.
\begin{itemize}
 \item \underline{Case 1}: The subsequent state is perfectly synced. Then, we 
must have that $\Delta\Phi \geq (j+1)b \geq b$.
 \item \underline{Case 2}: The subsequent state is also almost synced. Then, 
note that the control information received by some party must be invalid or 
maliciously corrupted. Moreover, since $\max\{\ell_A,\ell_B\}$ remains 
unchanged and $j$ can increase by at most 1, it follows that $\Delta\Phi \geq 
-b \geq -\cmal B$.
 \item \underline{Case 3}: The subsequent state is unsynced. Then, observe that 
the control information received by some party must be maliciously 
corrupted. Note that ${\ell^+}'' \geq \max\{\ell_A, \ell_B\} - B$, and 
${\ell^-}'' \leq 3B$. Moreover, $k_A'' = k_B'' = 1$. Therefore, if $\synca'' = 
\syncb''$, then
\begin{align*}
 \Delta\Phi &\geq -B(1+C_0 H(\epsilon)) - 3C_1 B + 2bC_2 - bC_4 
\\
 &\geq -\cmal B,
\end{align*}
while if $\synca'' \neq \syncb''$, then
\begin{align*}
 \Delta\Phi &\geq -B(1+C_0 H(\epsilon)) - 3C_1 B - 1.6bC_5 - bC_6\\
 &\geq -\cmal B,
\end{align*}
as desired.

\end{itemize}
\end{proof}

\begin{lemma}\label{lem:unsynctrans}
 Suppose the protocol is in an unsynced state at the beginning of an iteration. 
Then, the change in potential $\Phi$ over the course of the iteration behaves 
as follows, according to the control information received during the iteration:
\begin{enumerate}
 \item If the control information received by both parties is sound, then
$\Delta\Phi \geq b$.
 \item If the control information received by at least one party is invalid,
but neither party's control information is maliciously corrupted, then
$\Delta\Phi \geq -\cinv b$.
 \item If the control information received by at least one party is maliciously
corrupted, then $\Delta\Phi \geq -\cmal B$.
\end{enumerate}
\end{lemma}

\begin{proof}
We consider several cases, depending on the values of $k_A, k_B$ and what
transitions occur before the end of the iteration.
\begin{itemize}
   \item \underline{Case 1}: $k_A \neq k_B$.
  \begin{itemize}
    \item \underline{Subcase 1}: No transitions occur before the start of the
next iteration.
  \begin{enumerate}[label=\alph*.)]
    \item If the control information sent by both parties is sound or invalid, 
then note 
that $\Delta k_A = \Delta E_A \in \{0,1\}$ and $\Delta k_B = \Delta E_B \in 
\{0,1\}$. Also, at least one of $\Delta k_A$, $\Delta k_B$ must be 1, while
$\ell^+$, $\ell^-$, $\malAB$ remain unchanged. Moreover, the state
will remain an unsynced state with $k_A'' \neq k_B''$. Therefore,
\[
 \Delta\Phi \geq b(-0.8 C_5 + 0.9 C_5) \geq b.
\]

\item If at least one party's control information is maliciously corrupted and 
$k_A, k_B > 1$,
then note that the state at the beginning of the next iteration will also be
unsynced with $k_A'' \neq k_B''$. Also, observe that $\Delta k_A = \Delta k_B = 
1$, while $\ell^+, \ell^-$ remain unchanged. Thus,
\[
 \Delta\Phi \geq 2b(-0.8C_5) - 2C_7 B \geq -\cmal B.
\]

\item If at least one party's control information is maliciously corrupted and 
one of $k_A,
k_B$ is 1, then without loss of generality, assume $k_A = 1$ and $k_B > 1$.
Note that $k_B$ increases by 1. Also, if $k_A$ does not
increase, then $\ell^-$ can increase by at most $B$. Hence,
\[
 \Delta\Phi \geq -0.8bC_5 - 2C_7B - \max\{0.8b C_5, C_1 B\} \geq -\cmal B.
\]
\end{enumerate}

\item \underline{Subcase 2}: Only one of Alice and Bob undergoes a transition 
before the start of the next iteration. Without loss of generality, assume that
Alice makes the transition. Also, let
\begin{align}
 P_1 = \begin{cases}
      0.2 C_7 (k_A+1)B - (1+C_0 H(\epsilon) + C_1)k_A B \ &\text{if Alice has 
an MP trans.}\\ 0 \ &\text{otherwise}
     \end{cases} \label{eq:pdef}
\end{align}
Note that $P_1\geq 0$ for a suitable choice of constants $C_0$, $C_1$, $C_7$. 
Also observe that if $k_A\geq 3$, then
\begin{align}
E_A \leq \frac{1}{2}(k_A + 1) - 1 + 0.2\cdot \frac{1}{2}(k_A+1) = 0.6 k_A - 0.4 
\leq 0.7(k_A-1), \label{eq:eabd}
\end{align}
since an error transition did not occur when Alice's backtracking parameter was 
equal to $\frac{1}{2}(k_A+1)$, and an additional $\frac{1}{2}(k_A+1)-1$ 
iterations have occurred since then. Note that (\ref{eq:eabd}) also holds if 
$k_A < 3$ since it must be the case that $E_A = 0$.

\begin{enumerate}[label=\alph*.)]
 \item Suppose the control information sent by each party is sound. Then, note 
that $(k_A'', \synca'') \neq (k_B'', \syncb'')$. Moreover, if Alice's 
transition is a meeting point transition, then we must have $\malA \geq 
0.2(k_A+1)$, and the transition can cause Alice's transcript $T_A$ to be 
rewound by at most $k_A B$ bits, which implies that $\Delta\ell^- \leq k_A B$ 
and $\Delta\ell^+ \geq -k_A B$.

Thus, if $k_A, k_B > 1$, then by (\ref{eq:eabd}), we have
\begin{align*}
 \Delta\Phi &\geq 0.8b C_5 (k_A - 1)  - 0.9 b C_5 E_A + (-0.8b C_5 + 0.9 bC_5) 
+ P_1\\
&\geq 0.8 bC_5 (k_A - 1) - 0.9 bC_5 \cdot 0.7(k_A - 1) + 0.1bC_5\\
&\geq 0.27 b C_5\\
&\geq b,
\end{align*}
while if $k_A = 1$ and $k_B > 1$, then
\begin{align*}
 \Delta\Phi &\geq -0.8 b C_5 + 0.9 b C_5 + P_1\\
 &\geq 0.1 b C_5\\
 &\geq b.
\end{align*}
Finally, if $k_B = 1$, then $k_A > 1$ and so, by (\ref{eq:eabd}), we have
\begin{align*}
 \Delta\Phi &\geq 0.8 bC_5(k_A - 1) - 0.9 bC_5 E_A - bC_6 + P_1\\
 &\geq 0.8bC_5 (k_A - 1) - 0.9 bC_5 \cdot 0.7 (k_A - 1) - bC_6\\
 &\geq (0.17 C_5 - C_6)b\\
 &\geq b.
\end{align*}

\item Suppose the control information sent by at least one party is invalid, 
but neither party's control information is maliciously corrupted. Again, we 
note that if Alice's transition is a meeting point transition, 
then $\malA \geq 0.2(k_A+1)$ and $\Delta\ell^- \leq k_A B$ and $\Delta\ell^+ 
\geq -k_A B$.

First, suppose that $k_B = \syncb = 1$ and that Bob receives invalid control 
information. Then, note that $(k_A'', \synca'') = (k_B'', \syncb'') = (1,1)$. 
Thus, by (\ref{eq:eabd}),
\begin{align*}
 \Delta\Phi &\geq 0.8 bC_5 (k_A+1) - 0.9 bC_5 E_A + 2bC_2 - bC_4 + P_1\\
 &\geq 0.8 bC_5 (k_A+1) - 0.9 bC_5 \cdot 0.7(k_A-1) + 2bC_2 - bC_4\\
 &\geq (2C_2 - C_4 + 1.77 C_5)b\\
 &\geq -\cinv b.
\end{align*}

Next, suppose that $k_B = \syncb = 1$ but Bob receives sound information. Then, 
note that $(k_A'', \synca'') \neq (k_B'', \syncb'')$. Hence, by 
(\ref{eq:eabd}), 
\begin{align*}
 \Delta\Phi &\geq 0.8 bC_5 (k_A-1) - 0.9 bC_5 E_A - bC_6 + P_1\\
  &\geq 0.8bC_5 (k_A-1) - 0.9 bC_5 \cdot 0.7 (k_A - 1) - bC_6\\
  &\geq (0.17 C_5 - C_6)b\\
  &\geq -\cinv b.
\end{align*}

Finally, suppose that $(k_B, \syncb) \neq (1,1)$. Then, $\Delta k_B = \Delta 
E_B = 1$. Thus, if $k_A > 1$, then by (\ref{eq:eabd}),
\begin{align*}
 \Delta\Phi &\geq 0.8bC_5 (k_A-1) - 0.9 bC_5 E_A + (-0.8bC_5 + 0.9 bC_5) + P_1\\
 &\geq 0.8 b C_5 (k_A - 1) - 0.9 bC_5 \cdot 0.7 (k_A - 1) + 0.1bC_5\\
 &\geq 0.27 C_5 b\\
 &\geq -\cinv b,
\end{align*}
while if $k_A = 1$, Alice's transition must be an error transition and so,
\begin{align*}
 \Delta\Phi &\geq -0.8 b C_5 + 0.9 b C_5\\
 &= 0.1 C_5 b\\
 &\geq -\cinv b.
\end{align*}

\item Suppose the control information sent by at least one of the parties is
maliciously corrupted. If Alice's transition is a meeting point 
transition, then $\malA \geq 0.2(k_A+1) - 1$, and $T_A$ can be rewound up to 
at most $k_A B$ bits during the transition.

First, suppose that $(k_B'', \syncb'') \neq (1,1)$. Then, $\Delta k_B \leq 1$ 
and $\Delta \malB \leq 1$. Thus, by (\ref{eq:eabd}), we have
\begin{align*}
 \Delta\Phi &\geq 0.8 bC_5 (k_A - 1) - 0.9 bC_5 E_A - 0.8 bC_5 - C_7 B - 
bC_6 + (P_1 - C_7 B)\\
&\geq 0.8bC_5 (k_A - 1) - 0.9 bC_5 \cdot 0.7 (k_A - 1) - 0.8bC_5 - C_7 B - 
bC_6 - C_7 B\\
&\geq -(0.8 C_5 + C_6)b - 2C_7 B\\
&\geq -\cmal B.
\end{align*}

Next, suppose that $(k_B'', \syncb'') = (1,1)$. Then, since Bob does not 
undergo a transition, we have $k_B = \syncb = 1$. Also, the length of $T_B$ can 
increase by at most $B$ bits over the course of the next iteration. Hence,
\begin{align*}
 \Delta\Phi &\geq 0.8 bC_5 k_{AB} - 0.9 bC_5 E_{AB} - C_1 B + (P - C_7 B) + 
2bC_2 - bC_4\\
 &\geq 0.8 bC_5 (k_A + 1) - 0.9 bC_5 \cdot 0.7 (k_A - 1) - C_1 B - C_7 B + 
2bC_2 - bC_4\\
 &\geq (2C_2 - C_4 + 1.6 C_5)b - (C_1 + C_7)B\\
 &\geq -\cmal B.
\end{align*}
\end{enumerate}
  
\item \underline{Subcase 3}: Both Alice and Bob undergo transitions before the
start of the next iteration. Again, note that note
that $E_A \leq 0.7(k_A - 1)$, due to (\ref{eq:eabd}). Similarly, $E_B \leq 
0.7(k_B - 1)$. Also, we define $P_1$ as in (\ref{eq:pdef}) and define $P_2$ 
analogously:
\[
 P_2 = \begin{cases} 0.2 C_7(k_B+1)B - (1+C_0H(\epsilon)+C_1) k_B B\ &\text{if 
Bob has an MP trans.}\\
0\ &\text{otherwise}
\end{cases}.
\]
Observe that $P_1, P_2 \geq 0$ for a suitable choice of constants $C_0$, $C_1$, 
$C_7$.

First, suppose that no party receives maliciously corrupted control 
information. Then, note that if Alice undergoes a meeting point transition, 
then $\malA \geq 0.2 (k_A+1)$, and the transition can cause $T_A$ to be rewound 
by at most $k_A B$ bits. Similarly, if Bob undergoes a meeting point transition, 
then $\malB \geq 0.2 (k_B+1)$, and the transition can cause $T_B$ to be rewound 
by at most $k_B B$ bits. Thus, regardless of the types of transitions that Alice 
and Bob make, we have
\begin{align*}
 \Delta\Phi &\geq 0.8 bC_5 k_{AB} - 0.9 bC_5 E_{AB} + P_1 + P_2 + 2bC_2 - bC_4\\
 &\geq 0.8 bC_5 k_{AB} - 0.9 bC_5 \cdot 0.7 ((k_A-1) + (k_B-1)) 
+ 2bC_2 - bC_4\\
 &\geq (2C_2 - C_4 + 1.6 C_5)b\\
 &\geq b,
\end{align*}

Now, suppose some party receives maliciously corrupted control information. 
We instead have $\malA \geq 0.2 (k_A+1)-1$ and $\malB \geq 0.2(k_A+1)-1$. Thus,
\begin{align*}
 \Delta\Phi &\geq 0.8bC_5 k_{AB} - 0.9 bC_5 E_{AB} + (P_1-C_7 B) + (P_2 - C_7 
B) + 2bC_2 - bC_4\\
&\geq (2C_2 - C_4 + 1.6C_5)b - 2C_7 B\\
&\geq -\cmal B,
\end{align*}
as desired.
\end{itemize}
  
  \item \underline{Case 2}: $k_A = k_B = 1$.
\begin{itemize}
\item \underline{Subcase 1}: $\synca = \syncb = 1$. Then, note that if both 
parties receive sound control information, then $\synca'' = \syncb'' = 0$. Thus,
\[
 \Delta\Phi = -\Delta Z_1 = \frac{1}{2}bC_4 \geq b.
\]
On the other hand, if some party receives invalid control information but 
neither party receives maliciously corrupted control information, then note 
that either $\synca'' = \syncb'' = 1$, in which case,
\[
 \Delta\Phi = 0 \geq -\cinv b,
\]
or $\synca'' \neq \syncb''$, in which case,
\[
\Delta\Phi \geq -2bC_2 + bC_4 - 1.6bC_5 - bC_6 \geq -\cinv b.
\]
Finally, consider the case in which  some party receives maliciously corrupted 
information. Then, if $\synca'' = \syncb''$, note that $\Delta\ell^- \leq 2$. 
Thus, if the subsequent state is unsynced, then
\[
 \Delta\Phi \geq -2C_1 B \geq -\cmal B,
\]
while if the subsequent state is perfectly or almost synced, then
\[
 \Delta\Phi \geq -2bC_2 + bC_4 - (2s+1)b \geq -\cmal B.
\]

Otherwise, if $\synca'' \neq \syncb''$, then $\Delta\ell^- \leq 1$, and so,
\begin{align*}
 \Delta\Phi \geq -C_1 B -2bC_2 + bC_4 - 1.6bC_5 -bC_6 \geq -\cmal B.
\end{align*}

\item \underline{Subcase 2}: $\synca = \syncb = 0$. First, suppose both parties 
receive sound control information. Then, either both parties do not undergo any 
transitions, in which case,
\[
 \Delta\Phi \geq 2bC_2 + \frac{1}{2}bC_4 \geq b,
\]
or both parties undergo a meeting point transition, in which case the 
subsequent state is perfectly synced, and so,
\[
 \Delta\Phi \geq -2bC_2 + \frac{1}{2}bC_4 \geq b.
\]

Next, consider the case in which some party receives invalid control 
information, but neither party receives maliciously corrupted control 
information. Suppose, without loss of generality, that Alice receives invalid 
control information. Then, $k_A'' = \synca'' = 1$. Note that if $k_B'' = 2$, 
then
\[
 \Delta\Phi \geq -2bC_2 + \frac{1}{2}bC_4 - 2.4bC_5 \geq -\cinv b.
\]
Otherwise, if $k_B'' = 1$, then either the subsequent state is perfectly 
synced, in which case
\[
 \Delta\Phi \geq -2bC_2 + \frac{1}{2}bC_4 \geq -\cinv b,
\]
or the subsequent state is almost synced, in which case
\[
 \Delta\Phi \geq B(1+C_0 H(\epsilon)) - 2bC_2 + \frac{1}{2} bC_4 - b \geq 
-\cinv b,
\]
or the subsequent state is unsynced, in which case
\[
 \Delta\Phi \geq -\frac{1}{2}bC_4 \geq -\cinv b.
\]

Finally, consider the case in which some party receives maliciously corrupted 
control information. If $k_A'' = k_B'' = 2$, then
\[
 \Delta\Phi \geq 2bC_2 - 4C_7 B + \frac{1}{2}bC_4 \geq -\cmal B.
\]
On the other hand, if $k_A'' = k_B'' = 1$, then $\Delta\ell^- \leq 2$. Thus, if 
the subsequent state is unsynced, then
\[
 \Delta\Phi \geq -2(1+C_0 H(\epsilon) + C_1)B - \frac{1}{2}bC_4 \geq -\cmal B,
\]
while if the subsequent state is perfectly or almost synced, then
\[
 \Delta\Phi \geq -2(1+C_0 H(\epsilon) + C_1)B + \frac{1}{2} bC_4 - b \geq 
-\cmal B.
\]
If $k_A''\neq k_B''$, then without loss of generality, assume that $k_A'' = 2$ 
and $k_B'' = 1$. We then have
\[
 \Delta\Phi \geq -(1+C_0 H(\epsilon) + C_1)B -2bC_2 + \frac{1}{2}bC_4 - 2.4 
bC_5 - C_7 B \geq -\cmal B.
\]

\item \underline{Subcase 3}: $\synca \neq \syncb$. Without loss of generality, 
assume that $\synca=1$ and $\syncb=0$. 

First, suppose that neither party receives maliciously corrupted control 
information. Then, $k_A'' = \synca'' = k_B'' = \syncb'' = 1$. Thus, if the 
subsequent state is unsynced, then we have
\[
 \Delta\Phi \geq 1.6bC_5 + bC_6 + 2bC_2 - bC_4 \geq b,
\]
while if the subsequent state is perfectly or almost synced, then
\[
 \Delta\Phi \geq 1.6bC_5 + bC_6 - b \geq b.
\]

Next, suppose that some party receives maliciously corrupted control 
information. Note that $k_A'' = 1$. If $\synca = 1$ and $k_B'' = 2$, then 
$\Delta\ell^- \leq 1$, and so,
\[
 \Delta\Phi \geq -C_1 B - 0.8 bC_5 - C_7B + bC_6 \geq -\cmal B.
\]
If $\synca = 1$ and $k_B'' = 1$, then either the subsequent state is unsynced, 
in which case,
\[
 \Delta\Phi \geq -C_1 B - (1+C_0 H(\epsilon) + C_1) B + 0.8bC_5 + bC_6 + 2bC_2 
- bC_4 \geq -\cmal B,
\]
or the subsequent state is perfectly/almost synced, in which case,
\[
 \Delta\Phi \geq -C_1 B - (1+C_0 H(\epsilon) + C_1) B + 1.6bC_5 + bC_6 - 
(2s+1)b \geq -\cmal B.
\]

Finally, suppose $\synca = 0$. Then, note that
\[
 \Delta\Phi \geq -(1+C_0 H(\epsilon) + C_1)B - 0.8bC_5 - C_7 B \geq -\cmal B.
\]

\end{itemize}

\item \underline{Case 3}: The protocol is in an unsynced state, and $k_A = k_B >
1$.
\begin{itemize}
\item \underline{Subcase 1}: Suppose neither Alice nor Bob undergoes a
transition before the start of
the next iteration. Then, we have $\Delta k_A = \Delta k_B = 1$. If the
control information received by both parties is either sound or invalid, then 
we have
\[
 \Delta \Phi \geq 2 b C_2 \geq b.
\]
On the other hand, if some party's control information is maliciously 
corrupted, then
\[
 \Delta\Phi \geq 2bC_2 - 2bC_3 - 4BC_7 \geq -\cmal B.
\]

\item \underline{Subcase 2}: Suppose both Alice and Bob undergo a transition,
and suppose at least one of the transitions is a meeting point transition.

\begin{enumerate}[label=\alph*.)]
 \item Suppose ${\ell^-}'' = 0$ and $k_A + 1 = k_B + 1 \leq \frac{4\ell^-}{B}$.
Then, note that $\ell^+$ decreases by at most $k_A B = k_B B$. Thus,
\begin{align*}
 \Delta\Phi &\geq -k_A B (1 + C_0 H(\epsilon)) + C_1 \ell^- - 2C_2 b (k_A -
1) - C_4 b\\
&\geq -k_A B (1 + C_0 H(\epsilon)) + C_1\cdot \frac{B(k_A+1)}{4} - 2C_2 b (k_A
- 1) - C_4 b\\
&= k_A B \left(\frac{C_1}{4} - C_0 H(\epsilon) - \frac{2C_2 b}{B} - 1\right) +
\frac{C_1 B}{4} + (2C_2 - C_4) b\\
&\geq b.
\end{align*}

\item Suppose ${\ell^-}'' \neq 0$. Without loss of generality, assume that
Alice has made a meeting point transition. Note that if Alice has made an
incorrect meeting point transition, then it is clear that $\malA' \geq
0.2(k_A+1)$. On the other hand, if she has made a correct transition, then Bob
has made an incorrect transition, since ${\ell^-}'' \neq 0$, and so, $\malB'
\geq 0.2(k_A+1)$. Since $\malA' = \malB'$, it follows that $\malAB' \geq
0.4(k_A+1)$ in either case. Thus, if the control information in the current
round is not maliciously corrupted, then $\malAB \geq 0.4(k_A+1)$, and so,
\begin{align*}
 \Delta\Phi &\geq -k_A B(1+C_0 H(\epsilon)+C_1)  - 2C_2 b (k_A - 1) +
2 C_7 B \cdot 0.4(k_A + 1) - C_4 b\\
&\geq k_A B \left(0.8 C_7 - C_0 H(\epsilon) - C_1 - \frac{2C_2 b}{B} - 1\right)
+ (2C_2 - C_4) b + 0.8 C_7 B\\
&\geq b.
\end{align*}
Otherwise, if some party's control information in the current round is
corrupted, then $\malAB \geq 0.4(k_A+1) - 2$, and so,
\begin{align*}
 \Delta\Phi &\geq k_A B \left(0.8 C_7 - C_0 H(\epsilon) - C_1 - \frac{2C_2 b}{B}
- 1\right) + (2C_2 - C_4) b - 3.2 C_7 B\\
&\geq -\cmal B.
\end{align*}

\item Suppose that ${\ell^-}'' = 0$ but $k_A + 1 = k_B + 1 > \frac{4\ell^-}{B}$.
Then observe that there must have been at least
\begin{equation}
 \frac{1}{4}(k_A+1)-0.2\cdot \frac{1}{2}(k_A+1) - 0.2\cdot\frac{1}{2}(k_A+1) =
0.05 (k_A+1) \label{eq:malcorr}
\end{equation}
maliciously corrupted rounds among the past $k_A$ rounds. This is because there
were $\frac{1}{4}(k_A+1)$ iterations taking place as Alice's backtracking 
parameter increased from $\frac{1}{4}(k_A+1)$ to $\frac{1}{2}(k_A+1)$, of which 
at most $0.2\cdot\frac{1}{2}(k_A+1)$ iterations could have had invalid control
information for Alice, and at most $0.2\cdot\frac{1}{2}(k_A+1)$ iterations
could have had sound control information for Alice (since Alice did not
undergo a meeting point transmission when her backtracking parameter reached
$\frac{k_A+1}{2}$). Thus, $\malAB\geq 2\cdot 0.05(k_A+1) = 0.1(k_A+1)$ and
so, 
\begin{align*}
 \Delta\Phi &\geq -k_A B (1 + C_0 H(\epsilon)) - 2b C_2 (k_A-1) + C_7 B
\cdot \malAB - C_4 b\\
&\geq k_A B\left( 0.1 C_7 - C_0 H(\epsilon) - \frac{2C_2 b}{B} - 1\right) +
(2C_2-C_4)b + 0.1C_7 B\\
&\geq b.
\end{align*}

\end{enumerate}

\item \underline{Subcase 3}: Suppose both Alice and Bob undergo error
transitions. Then, $E_A' \geq 0.2 (k_A+1)$ and $E_B' \geq 0.2
(k_B+1) = 0.2(k_A+1)$. Therefore, if both parties receive sound control 
information, then $E_A, E_B \geq 0.2(k_A+1)$, and so,
\begin{align*}
 \Delta\Phi &\geq C_3 bE_{AB} - 2C_2 b(k_A-1) - C_4 b\\
&\geq C_3 b(0.4k_A + 0.4) - 2C_2 b(k_A-1) - C_4 b\\
&\geq (0.4 C_3 - 2C_2)k_A b + (2C_2 + 0.4 C_3 - C_4)b\\
&\geq (0.8C_3 - C_4)b\\
&\geq b.
\end{align*}
On the other hand, if some party receives invalid or maliciously corrupted 
control information, then $E_A, E_B \geq 0.2(k_A+1) - 1$, and so,
\begin{align*}
 \Delta\Phi &\geq C_3 bE_{AB} - 2C_2 b(k_A-1) - C_4 b\\
 &\geq (0.4C_3 - 2C_2)k_A b + (2C_2 - 1.6 C_3 - C_4)b\\
 &\geq (-1.2 C_3 - C_4)b\\
 &\geq -\cinv b.
\end{align*}

\item \underline{Subcase 4}: Suppose only one of Alice and Bob undergoes a 
transition
before the next iteration. Without loss of generality, assume Alice undergoes
the transition. 

\begin{enumerate}[label=\alph*.)]
 \item Suppose the transition is an error transition. If both parties' 
control information is sound, then observe that $E_A
\geq 0.2(k_A+1)$. Thus,
\begin{align*}
\Delta\Phi &\geq -2bC_2 k_A + bC_3 E_A  - 0.8 bC_5 (k_A + 2)\\
&\geq -2bC_2 k_A + b C_3 (0.2 k_A + 0.2) - 0.8 bC_5 (k_A+2)\\
&\geq k_A b (0.2 C_3 - 0.8 C_5 - 2C_2) + (0.2 C_3 - 1.6C_5)b\\
&\geq b.
\end{align*}
Otherwise, if some party's control information is invalid, but neither 
party's control information is maliciously corrupted, then $E_A \geq 
0.2(k_A+1)-1 = 0.2k_A - 0.8$, and so,
\begin{align*}
\Delta\Phi &\geq -2bC_2 k_A + bC_3 E_A  - 0.8 bC_5 (k_A + 2)\\
&\geq -2bC_2 k_A + b C_3 (0.2 k_A - 0.8) - 0.8 bC_5 (k_A+2)\\
&\geq k_A b (0.2 C_3 - 0.8 C_5 - 2C_2) - (0.8 C_3 + 1.6C_5)b\\
&\geq -\cinv b.
\end{align*}
Finally, if some party's control information is maliciously corrupted, then 
again, we have $E_A \geq 0.2 k_A - 0.8$. Thus,
\begin{align*}
\Delta\Phi &\geq -2bC_2 k_A + bC_3 E_A  - 0.8 bC_5 (k_A + 2) - C_7 B\\
&\geq  k_A b (0.2 C_3 - 0.8 C_5 - 2C_2) - (0.8 C_3 + 1.6C_5)b - C_7 B\\
&\geq -\cmal B.
\end{align*}

\item Suppose the transition is a meeting point transition. Then, since
only one of the two players is transitioning, either (1.) Alice is incorrectly
transitioning, meaning that $\malA', \malB' \geq 0.2 (k_A+1)$, or  (2.)  Bob
should have also been transitioning, meaning that $\malA',
\malB' \geq\frac{1}{2}(k_A+1) - 0.2 (k_A+1) - 0.2(k_A + 1) \geq 0.1(k_A+1)$.
Either way, $\malA', \malB' \geq 0.1(k_A+1)$.

Hence, if neither party's control information in the current round is 
maliciously corrupted, then $\malA,\malB \geq 0.1(k_A+1)$, and so,
\begin{align*}
\Delta\Phi &\geq -2bC_2 k_A - 0.8b C_5 (k_A+2) + 2C_7 B \cdot \malA + C_7 B
\cdot \malB\\
&\quad  - k_A B(1 + C_0H(\epsilon)+C_1) \\
&\geq -2bC_2 k_A - 0.8b C_5 (k_A+2) + 0.3C_7 B (k_A+1) - k_A B(1 +
C_0H(\epsilon)+C_1) \\
&\geq k_A B \left(0.3 C_7 - C_1 - C_0 H(\epsilon) - 2C_2 \frac{b}{B} - 0.8C_5
\frac{b}{B} - 1\right) - 1.6bC_5 + 0.3 C_7 B\\
&\geq b.
\end{align*}
Otherwise, if there is maliciously corrupted control information in the current
round, then $\malA,\malB\geq 0.1(k_A+1)-1 = 0.1 k_A - 0.9$, and so,
\begin{align*}
\Delta\Phi &\geq -2bC_2 k_A - 0.8b C_5 (k_A+2) + 2C_7 B \cdot \malA + C_7 B
\cdot \malB - C_7 B\\
&\quad  - k_A B(1 + C_0H(\epsilon)+C_1) \\
&\geq -2bC_2 k_A - 0.8b C_5 (k_A+2) + 3C_7 B (0.1k_A-0.9) - C_7 B\\
&\quad - k_A B(1 + C_0H(\epsilon)+C_1) \\
&\geq k_A B \left(0.3 C_7 - C_1 - C_0 H(\epsilon) - 2C_2 \frac{b}{B} - 0.8C_5
\frac{b}{B} - 1\right) - 1.6bC_5 - 2.7 C_7 B\\
&\geq -\cmal B,
\end{align*}
as desired.
\end{enumerate}
\end{itemize}
\end{itemize}
\end{proof}

Now, we are ready to prove the main theorem of the section, which implies 
Theorem~\ref{thm:mainoblivious} for the choice $\epsilon' = \epsilon^2$.
\begin{theorem}\label{thm:nonrate}
 For any sufficiently small $\epsilon > 0$ and $n$-round interactive protocol 
$\Pi$ with average message length $\ell = \Omega(1/\epsilon'^3)$, the protocol 
$\pienc$ given in Figure~\ref{fig:oblivious} successfully simulates $\Pi$, 
with probability $1-2^{-\Omega(\epsilon'^2 \Niter)}$, over an oblivious 
adversarial channel with an $\epsilon$ error fraction while 
achieving a communication rate of $1-\Theta(\epsilon\log(1/\epsilon)) = 
1-\Theta(H(\epsilon))$.
\end{theorem}
\begin{proof}
 Recall that $\piblk$ has $n'$ rounds, where $n' = n(1 +O(\epsilon'))$. Let 
$\Nmal$ be the number of iterations of $\pienc$ in which some party's control 
information is maliciously corrupted. Moreover, let $\Ninv$ be the number of 
iterations in which some party's control information is invalid but neither 
party's control information is maliciously corrupted. Finally, let $\Nsound$ be 
the number of iterations starting at an unsynced or almost synced state such 
that both parties receive sound control information.

Now, by Lemma~\ref{lem:malbound}, we know that with probability
$1-2^{-\Omega(\epsilon'^2 \Niter)}$, $\Nmal 
= O(\epsilon'^2 \Niter)$. Also, by Lemma~\ref{lem:invbound}, $\Ninv = 
O(\epsilon \Niter)$ with probability $1-2^{-\Omega(\epsilon' \Niter)}$. Recall
that the total number of data bits that can be corrupted 
by the adversary throughout the protocol is at most $\epsilon b\Niter$. Since 
$\Niter = \Nsound + \Ninv + \Nmal$, Lemmas~\ref{lem:synctrans}, 
\ref{lem:asynctrans}, and \ref{lem:unsynctrans} imply that at the end of the 
execution of $\pienc$, the potential function $\Phi$ satisfies
\begin{align*}
 \Phi &\geq b\Nsound - C \epsilon b \Niter\log(1/\epsilon) - \cinv b \Ninv - 
\cmal B \Nmal\\
&= b(\Niter-\Ninv-\Nmal) - C \epsilon b \Niter\log(1/\epsilon) - \cinv b \Ninv 
- \cmal B \Nmal\\
&= b\Niter - C\epsilon b \Niter\log(1/\epsilon) - (\cinv + 1)b\Ninv - (\cmal B 
+ 
b)\Nmal \\
&= b\Niter - C \epsilon b \Niter\log(1/\epsilon) - O(\epsilon) 
\cdot (\cinv+1) b\Niter - O(\epsilon'^2) \cdot (\cmal B + b) \Niter\\
&= b\Niter (1 - O(\epsilon) \cdot (\cinv+1) - O(\epsilon'^2) \cdot (\cmal s + 
1) - C\epsilon\log(1/\epsilon))\\
&= b\Niter (1 - O(\epsilon\log(1/\epsilon)))\\
&= b\cdot \frac{n'}{b}(1+\Theta(\epsilon\log(1/\epsilon)))\\
&\geq n' (1+C_0 H(\epsilon)) + (C_0 + 1)B.
\end{align*}
Now, in order to complete the proof, it suffices to show that $\ell^+ \geq 
n'$. We consider several cases, based on the ending state:
\begin{itemize}
 \item If the ending state is perfectly synced, then note that $jb - 
C\cdot\curerr\cdot \log(1/\epsilon) \leq 2B$. Thus,
\[
 \ell^+ \geq \frac{\Phi - 2B}{1 + C_0 H(\epsilon)} \geq n'.
\]

\item If the ending state is almost synced, then note that
\[
 \ell^+ \geq \frac{\Phi}{1+C_0 H(\epsilon)} - B \geq n'.
\]

\item If the ending state is unsynced and $(k_A, \synca) = (k_B, \syncb)$, then 
first consider the case $k_A = k_B = 1$. In this case,
\[
 \Phi \leq \ell^+ (1+C_0 H(\epsilon)) + 2bC_2,
\]
and so,
\[
 \ell^+ \geq \frac{\Phi - 2bC_2}{1+C_0 H(\epsilon)} \geq n'.
\]

Now, consider the case $k_A = k_B \geq 2$. Note that either $\ell^- \geq 
\frac{B}{4}(k_A+1)$ or
\begin{align*}
 \malAB &\geq 2\cdot\malA \geq 2 \left(\frac{1}{2}\widt{k}_A - 0.2 \widt{k}_A - 
0.2\widt{k}_A\right) \geq 0.2 \widt{k}_A \geq 0.1 (k_A+1)
\end{align*}
(see (\ref{eq:malcorr})). If the former holds, then
\begin{align*}
 \Phi &\leq \ell^+ (1+C_0 H(\epsilon)) - C_1 \ell^- + bC_2 k_{AB}\\
 &\leq \ell^+ (1+ C_0 H(\epsilon)) - C_1 \cdot \frac{B}{4}(k_A+1) + 2bC_2 k_A\\
 &\leq \ell^+ (1+C_0 H(\epsilon)).
\end{align*}
Otherwise, if the latter holds, then
\begin{align*}
 \Phi &\leq \ell^+ (1+C_0 H(\epsilon)) + bC_2 k_{AB} - 2C_7 B \malAB\\
 &\leq \ell^+ (1+C_0 H(\epsilon)) + 2bC_2 k_A - 2C_7 B(0.1(k_A+1))\\
 &\leq \ell^+ (1+C_0 H(\epsilon)).
\end{align*}
Either way,
\[
 \ell^+ \geq \frac{\Phi}{1+C_0 H(\epsilon)} \geq n'.
\]

\item If the ending state is unsynced and $k_A \neq k_B$, then consider the 
following. Note that if $k_A = 1$, then $E_A = 0 \leq 0.6 k_A - 0.4$. On the 
other hand, if $k_A \geq 2$, then
\begin{align*}
E_A &\leq 0.2 \widt{k}_A + (k_A -\widt{k}_A)\\
&= k_A - 0.8\widt{k}_A\\
&\leq k_A - 0.8\left(\frac{k_A+1}{2}\right)\\
&\leq 0.6 k_A - 0.4.
\end{align*}
Either way, $E_A \leq 0.6 k_A - 0.4$. Similarly, $E_B \leq 0.6 k_B - 0.4$. Thus,
\begin{align*}
 \Phi &\leq \ell^+ (1+C_0 H(\epsilon)) + bC_5(-0.8 k_{AB} + 0.9 
E_{AB})\\
&\leq \ell^+ (1+C_0 H(\epsilon)) + bC_5 (-0.8 k_{AB} + 0.9 ((0.6 k_A-0.4) + 
(0.6 k_B-0.4)))\\
&\leq \ell^+ (1+C_0 H(\epsilon)).
\end{align*}
Thus,
\[
 \ell^+ \geq \frac{\Phi}{1+C_0 H(\epsilon)} \geq n'.
\]
\end{itemize}
\end{proof}
\noindent Finally, we prove Theorem~\ref{thm:rateless}.
\begin{proof}
 Consider the same protocol $\pienc$ as in Theorem~\ref{thm:nonrate}, except 
that we discard the random string exchange procedure at the beginning of the 
protocol. Since Alice and Bob have access to public shared randomness, they can 
instead initialize $\shrand$ to a common random string of the appropriate 
length and continue with the remainder of $\pienc$. Moreover, in this case, 
$\epsilon'$ is a parameter that is set as part of the input. Then, it is clear 
that the analysis of Theorem~\ref{thm:nonrate} still goes through. In this 
case, we have that the total number of rounds is
\[
 \Niter\, b' = \frac{n'b'}{b}(1+O(\epsilon\log(1/\epsilon))) = n(1 + 
O(H(\epsilon)) + O(\epsilon'\,\mathrm{polylog}(1/\epsilon'))),
\]
while the success probability is $1 - 2^{-\Omega(\epsilon'^2 \Niter)} = 1 - 
2^{-\Omega(\epsilon'^3 n)}$, as desired.
\end{proof}

\begin{remark}
 It is routine to verify that the constants $C_0, C_1, C_2, C_3, C_4, 
C_5, C_6, C_7, \cinv, \cmal, C, D > 0$ can be chosen appropriately such that 
the relevant inequalities in Lemmas~\ref{lem:synctrans}, \ref{lem:asynctrans}, 
\ref{lem:unsynctrans}, and Theorem~\ref{thm:nonrate} all hold.
\end{remark}

\shortOnly{\newpage
\appendix

\section{Appendix}
\subsection{Related Works}\label{app:relworks}
\RelatedWorkSection

\subsection{Trivial Scheme for Non-Adaptive Protocols with Minimum Message 
Length Under Random Errors} \label{app:trivialscheme}
\TrivialScheme

\subsection{Preliminaries} \label{app:prelim}
\Prelim
\subsubsection{Communication Channels} \label{subsec:commchannel}
\CommChannels

\subsection{Proofs} \label{app:proofs}
\ProofofBlocking
\ProofofRatelessCode

\subsection{Encoding and Decoding Scheme for Control Information} 
\label{app:encscheme}
\EncodingScheme}

\newpage
\bibliographystyle{alpha}
\bibliography{references}

\end{document}